\documentclass[twocolumn]{autart}    
\usepackage[authoryear]{natbib}
\usepackage{amsmath}
\usepackage{color}
\usepackage[dvipsnames]{xcolor}
\usepackage{algorithm}
\usepackage{algpseudocode}
\usepackage{latexsym, amsbsy, amsopn, amstext}
\usepackage{soul}
\usepackage{mathrsfs}

\DeclareMathOperator{\lcm}{lcm}
\DeclareMathOperator{\Col}{Col}

\definecolor{c1}{HTML}{dF7e8a}
\definecolor{c2}{HTML}{9dcbFF}
\definecolor{c3}{HTML}{528e86}
\definecolor{c4}{HTML}{b7ced2}
\definecolor{c5}{HTML}{a19a90}
\definecolor{c6}{HTML}{eae1eF}
\definecolor{c7}{HTML}{0b8eab}
\definecolor{c8}{HTML}{aFca54}
\definecolor{c9}{HTML}{b8aFc9}

\definecolor{d1}{HTML}{de9602}
\definecolor{d2}{HTML}{acc864}
\definecolor{d3}{HTML}{b0d3bf}
\definecolor{d4}{HTML}{F2F2EF}

\definecolor{green}{rgb}{0.1,0.7,0.1}
\newcommand{\green}{\color{green}}

\newtheorem{them}{Theorem}
\newtheorem{dfn}{Definition}
\newtheorem{lemm}{Lemma}
\newtheorem{prp}{Proposition}
\newtheorem{exa}{Example}
\newtheorem{remr}{Remark}

\usepackage{amsmath,graphicx,tikz,enumerate}    
\usetikzlibrary{decorations.markings}
\usepackage{amssymb}
\usetikzlibrary{positioning,shapes,fit,calc}
\usetikzlibrary{shapes,shapes.geometric,arrows,fit,calc,positioning,automata}
\usepackage{hyperref} 
\usepackage{graphicx}
\usepackage{subcaption}



\begin{document}

\begin{frontmatter}
	\title{Structures of {\em M}-Invariant Dual Subspaces with Respect to a Boolean Network \thanksref{footnoteinfo}}    
	\thanks[footnoteinfo]{This work is supported by National Natural Science Foundation of China (Nos. 12071370 and 12131013). The material in this paper was not presented at any conference.} 
	
	\author[NWPU]{Dongyao Bi}\ead{bdy@mail.nwpu.edu.cn},
	\author[NWPU]{Lijun Zhang\corauthref{1}}\ead{zhanglj7385@nwpu.edu.cn},
	\author[Cagliari]{Kuize Zhang}\ead{kuize.zhang@unica.it},
	\author[NWPU1]{Shenggui Zhang}\ead{sgzhang@nwpu.edu.cn}
	
	\corauth[1]{Corresponding author.}
	\address[NWPU]{School of Marine Science and Technology, Northwestern Polytechnical University, Xi'an, Shaanxi 710072, P.R. China}
	\address[Cagliari]{Department of Electrical and Electronic Engineering, University of Cagliari, Cagliari 09123, Italy}
	\address[NWPU1]{School of Mathematics and Statistics, Northwestern Polytechnical University, Xi'an, Shaanxi 710129, P.R. China  }

	
	\begin{abstract}
	This paper presents the following research findings on Boolean networks (BNs) and their dual subspaces. 
	First, we establish a bijection between the dual subspaces of a BN and the partitions of its state set. Furthermore, we demonstrate that a dual subspace is $M$-invariant if and only if the associated partition is equitable (i.e., for every two cells of the partition, every two states in the former have the same number of out-neighbors in the latter) for the BN's state-transition graph (STG). Here $M$ represents the structure matrix of the BN. 
Based on the equitable graphic representation, we provide, for the first time, a complete structural characterization of the smallest $M$-invariant dual subspaces generated by a set of Boolean functions. Given a set of output functions, we prove that a BN is observable if and only if the partition corresponding to the smallest $M$-invariant dual subspace generated by this set of functions is trivial (i.e., all partition cells are singletons). Building upon our structural characterization, we also present a method for constructing output functions that render the BN observable.
	\end{abstract}
	
\end{frontmatter}

\section{Introduction}
Boolean networks (BNs), initially introduced by Kauffman in 1969 \citep{kau69}, have emerged as a highly effective approach for modeling and analyzing genetic regulatory networks. To model the external inputs and their impact on the system's outputs, BNs have been extended to Boolean control networks (BCNs) \citep{Datta03, Ideker01}. The concept of the semi-tensor product (STP) of matrices, proposed by Cheng in 2001 \citep{bche11,che12}, has provided an algebraic framework for handling BNs and BCNs. 

The state space of a BN consists of all its Boolean state vectors. The set of all Boolean functions on the state space is called the dual (state) space. A dual subspace generated by a given set of Boolean functions is defined as the set of all Boolean functions that take these given Boolean functions as arguments. Hence generally a dual subspace has nothing to do with a BN's dynamics. Consider a BN, where $M$ is the structure matrix of the BN under the STP framework \citep{bche11,che12}, a dual subspace is $M$-invariant if, for every Boolean function $f$ within this dual subspace, the Boolean function resulted from applying the BN dynamics to the arguments of $f$ also belongs to the same dual subspace \citep{che21}. The evolution of Boolean functions that generate an $M$-invariant dual subspace can induce a dual system of the original BN.
The properties of a dual system induced by $M$-invariant dual subspaces were investigated in \citep{che21}.
\citep{che21} also designed an algorithm to compute the smallest $M$-invariant dual subspace generated by a given set of Boolean functions.
This paper provides a complete structural characterization of the smallest $M$-invariant dual subspace from a graph-theoretic perspective.

When dealing with a large-scale BN (BCN), the size of the structure matrix for the entire network becomes immense, making it impractical to compute within a reasonable time. However, the dual systems derived from a BN's $M$-invariant dual subspaces often exhibit compactness and still carry valuable information from the original BN. For example, in a BN with output functions, the minimal realization of the BN refers to the dynamic equation of a reduced dual subspace that contains the original output functions. The reduced dual subspace is the smallest $M$-invariant dual subspace generated by these output functions. Like the concept of minimal realization in control theory, an $M$-invariant dual subspace preserves essential properties of the original BN (BCN) while filtering out redundant information related to state transitions. The original BN (BCN) structure can be partly revealed in the dual dynamics. \citep{che21}. As a result, an algorithm was proposed in \citep{che21} to compute the smallest $M$-invariant dual subspace containing a given set of Boolean functions. For recent works on $M$-invariant dual subspaces, the reader is referred to, e.g., \citep{li23,liyf23}.

The main contributions of this paper are threefold. First, to investigate the attributes of $M$-invariant dual subspaces, we establish a bijection between partitions of the BN's state set and its dual subspaces.
We demonstrate that such a dual subspace of a BN is $M$-invariant if and only if (iff) the corresponding partition of the BN's state-transition graph (STG) is equitable (i.e., every two states in the same partition cell have an equal number of out-neighbors in any partition cell). Furthermore, the quotient digraph of this equitable partition can be utilized to describe the dual dynamics of the equivalence classes derived from this $M$-invariant dual subspace.
Second, 
we provide, for the first time, a complete structural characterization of the smallest $M$-invariant dual subspaces generated by a set of Boolean functions utilizing the equitable partition representation for the $M$-invariant dual subspaces. 
Third, we demonstrate that the unobservable subspace of a BN is the smallest $M$-invariant dual subspace generated by its output functions. We conclude that a BN with given output functions is observable iff the partition corresponding to the unobservable subspace is trivial (i.e., all partition cells are singletons). Building upon the structures of the smallest $M$-invariant dual subspaces generated by various output functions, we finally introduce a method for constructing observable output functions (i.e., output functions that make the BN observable). 

The remainder of this paper is organized as follows: Section \ref{II} surveys necessary results in graph theory and STP. Section \ref{III} presents the main findings of the current paper.
Section \ref{IV} examines the unobservable subspace of a BN and proposes methods for constructing output functions to render a given BN observable utilizing the aforementioned structures.

\section{Preliminaries}\label{II}
\subsection{Basic knowledge in graph theory}
\subsubsection{Basic concepts and notations \label{basic in graph}}
A digraph is denoted by ${\mathcal G}=(V, E)$, where $V=\{v_1,v_2,\ldots,v_n\}$ and $E\subseteq V\times V$ represent the vertex set and the edge set, respectively. For $(v_{i}, v_{j})\in E$, we refer to its two ends $v_{i}$ and $v_{j}$ as the tail and head of the edge, respectively. The ends of an edge are said to be adjacent to each other and incident to the edge. If $(v_{i}, v_{j})\in E$ (where $v_{i}$ and $v_{j}$ are not necessarily distinct), then $v_{i}$ is an in-neighbor of $v_{j}$, and $v_{j}$ is an out-neighbor of $v_{i}$. We define the in-neighbor (out-neighbor) set of $v_i$ as $N_{in}(v_i)$ ($N_{out}(v_i)$). The in-degree (out-degree) of $v_i$, denoted as $d_{in}(v_i)$ ($d_{out}(v_i)$), is the cardinality of $N_{in}(v_i)$ ($N_{out}(v_i)$).
We say $v_{j}$ is reachable from $v_{i}$ if there exists a path from $v_{i}$ to $v_{j}$.
The distance from $v_{i}$ to $v_{j}$, denoted by $\operatorname{dist}(v_{i},v_{j})$, is the length of the shortest paths from $v_{i}$ to $v_{j}$. 
Let $\operatorname{dist}^*_{in}(v_i)=\max\{\operatorname{dist}(v_{j},v_{i})|v_i\text{ is reachable from } v_j\}$. 
We denote by $N_{in}(v_{i}, k)$ the vertex set $\{v_{j}|\operatorname{dist}(v_{j},v_{i})=k\}$, where $k\in\{1,\ldots,\operatorname{dist}^*_{in}(v_{i})\}$. A loop, also called a self-loop, is an edge whose ends coincide. Specifically, if $v_{i}$ is incident with a loop, then $N_{in}(v_{i})=\{v_i\}\cup N_{in}(v_{i}, 1)$. Otherwise, $N_{in}(v_{i})=N_{in}(v_{i}, 1)$.
Similarly, we define $\operatorname{dist}^*_{out}(v_i)$ as $\max\{\operatorname{dist}(v_{i},v_{j})|v_j\text{ is reachable from } v_i\}$. Let $N_{out}(v_{i}, k)=\{v_{j}|\operatorname{dist}(v_{i},v_{j})=k\}$, where $k\in\{1,\ldots,\operatorname{dist}^*_{out}(v_{i})\}$.

The underlying graph of ${\mathcal G}$ is an undirected graph on the same vertex set; for each directed edge in ${\mathcal G}$, there exists an undirected edge with the same ends. A digraph is weakly connected if an undirected path exists in its underlying graph between any pair of vertices. Every digraph can be expressed uniquely (up to order) as a disjoint union of maximal weakly connected digraphs, which are called the components of ${\mathcal G}$.
For more information on graph theory, the reader is referred to \citep{bon08}.

A weighted digraph is a digraph with a weight function $w: E\rightarrow \mathbb{R}$ that assigns a weight $w((v_i,v_j))$ to each edge $(v_i,v_j)\in E$. For a given weighted digraph ${\mathcal G}$ (if ${\mathcal G}$ is unweighted, then edge weights are uniformly set to 1), the corresponding adjacency matrix $A({\mathcal G})$ is an $n\times n$ matrix defined as
$$
[A({\mathcal G})]_{ij}=\left\{\begin{array}{ll}
w((v_j,v_i)) ,&  (v_j,v_i)\in E;\\
0, & \text { otherwise. }\end{array}\right.
$$

\subsubsection{Graph partitions}
\label{section2.1.2}
For a digraph ${\mathcal G}=(V,E)$ with $n$ vertices and a given integer $1\leq k\leq n$, we call $\pi=\{C_1,C_2,\ldots,C_k\}$ a $k$-partition of $V$ if $\pi$ is a family of nonempty disjoint subsets of $V$ and $\cup_{i=1}^{k} C_i=V$. Each $C_i$ is referred to as a partition cell, where $1\leq i\leq k$. A partition is considered nontrivial if it contains at least one non-singleton cell; otherwise, it is trivial. The {\em characteristic matrix} $P(\pi)\in\{0,1\}^{n\times k}$ of the partition $\pi$ is defined as follows:
\begin{equation*}
[P]_{ij} =\left\{\begin{array}{l}1, \text { if }~v_i \in C_{j}; \\ 0, \text { otherwise, }\end{array} \quad 1 \leqslant i \leqslant n, 1 \leqslant j \leqslant k.\right.
\end{equation*}

\begin{dfn}\citep{agu17}{\label{d2.1}}
Let ${\mathcal G}=(V,E)$ be a weighted digraph with adjacency matrix $A({\mathcal G})$. A partition $\pi$ of $V$ is equitable with respect to ${\mathcal G}$ if for all pairs of partition cells $(C_i,C_j)$, where $i,j=1,\ldots,k$, and for all vertices $v_s, v_t \in C_i$
\begin{equation}
	\label{2.1.2-1}
	\sum_{v_k\in N_{out}(v_s)\cap C_j}[A({\mathcal G})]_{ks}=\sum_{v_k\in N_{out}(v_t)\cap C_j}[A({\mathcal G})]_{kt}.
\end{equation}
Particularly, when ${\mathcal G}$ is unweighted, $(\ref{2.1.2-1})$ degenerates to
\begin{equation}
	\label{def2.1-2}
	|N_{out}(v_s)\cap C_j|=|N_{out}(v_t)\cap C_j|.
\end{equation}
\end{dfn}

\begin{remr}
For unweighted digraphs, (\ref{def2.1-2}) implies that vertices within the same cell have an equal number of out-neighbors in any given cell.Fig. \ref{fig2.1} illustrates three examples of equitable partitions and their quotient digraphs. 
\end{remr}

For an equitable partition $\pi=\{C_1,C_2,\ldots,C_k\}$ of $V$ with respect to ${\mathcal G}$, the quotient digraph ${\mathcal G} / \pi$ of ${\mathcal G}$ over $\pi$ has the vertex set $V({\mathcal G} / \pi)=\left\{c_{1}, c_{2}, \ldots, c_{k}\right\}$ and edge set $E({\mathcal G} / \pi)=\left\{(c_{i}, c_{j}) \mid \exists v_s\in C_i:N_{out}(v_s)\cap C_j\neq \emptyset \right\}$; and the weight of $(c_{i}, c_{j})$ is $\sum_{v_k\in N_{out}(v_s)\cap C_j}[A({\mathcal G})]_{ks}, \forall v_s\in C_i$. With a slight abuse of notation, we refer to an equitable partition of $V$ with respect to ${\mathcal G}$ as an equitable partition of ${\mathcal G}$.

\begin{lemm}\citep{car07}{\label{l2.2}}
Let ${\mathcal G}$ be a digraph. A partition $\pi$ is equitable iff there exists a matrix $H$ satisfying $P^{\top}A=H P^{\top}$, where $P$ is the characteristic matrix of $\pi$ and $A$ is the adjacency matrix of ${\mathcal G}$. Moreover, if $\pi$ is equitable, then $H$ is exactly the adjacency matrix of the quotient digraph ${\mathcal G}/\pi$.
\end{lemm}

In Lemma \ref{l2.2}, $[P^{\top}A]_{ij}$ is the sum of the weights of the edges originating from $v_j$ and terminating in $C_i$, i.e., $\sum_{v_k\in N_{out}(v_j)\cap C_i}[A({\mathcal G})]_{kj}$. And $[HP^{\top}]_{ij}$ is equal to $[H]_{ik}$, where $C_k$ denotes the cell containing $v_j$. 

\begin{dfn}{\label{d2.3}}
A partition $\pi_1$ is said to be finer than $\pi_2$ if each cell of $\pi_2$ can be expressed as the union of some cells in $\pi_1$, denoted by $\pi_2\preceq\pi_1$. In this situation, we also call $\pi_2$ coarser than $\pi_1$.
\end{dfn}


\begin{dfn}\citep{che12}{\label{d2.3-1}}
Let $\Pi$ be the set of all the partitions of $V$. Consider $S\subset\Pi$.
\begin{itemize}
	\item[(i)] $\pi\in \Pi$ is an upper bound (a lower bound) of $S$ if $\pi'\preceq \pi$ ($\pi\preceq \pi'$) for all $\pi'\in S$.
	\item[(ii)] $\pi\in \Pi$ is the least upper bound of $S$, also the join of $S$, (denoted by $\pi=\sqcup S$), if $\pi$ is an upper bound of $S$, and for any other upper bound $\pi'$ of $S$, we have $\pi\preceq \pi'$.
	\item[(iii)] $\pi\in \Pi$ is the greatest lower bound of $S$, also the meet of $S$, (denoted by $\pi=\sqcap S$), if $\pi$ is a lower bound of $S$, and for any other lower bound $\pi'$ of $S$, we have $\pi'\preceq \pi$.
\end{itemize}
\end{dfn}

\begin{exa}\label{e2.4}

Fig. \ref{fig1-a} shows an unweighted digraph ${\mathcal G}$ with four vertices. Its adjacency matrix is 
\begin{equation*}
    A({\mathcal G})=\left[\begin{array}{cccc}
1 & 1 & 1 & 0\\
0 & 0 & 0 & 0\\
0 & 0 & 0 & 1\\
0 & 0 & 0 & 0
\end{array}\right].
\end{equation*}
Consider a partition $\pi_1:=\{C_1:=\{v_1,v_2\},C_2:=\{v_3\},C_3:=\{v_4\}\}$, its characteristic matrix is
\begin{equation*}
    P(\pi_1)=\left[\begin{array}{ccc}
1 & 0 & 0 \\
1 & 0 & 0 \\
0 & 1 & 0 \\
0 & 0 & 1
\end{array}\right].
\end{equation*}
For the only non-singleton cell $C_1$, $|N_{out}(v_1)\cap C_1|=|N_{out}(v_2)\cap C_1|=1$. According to Definition \ref{d2.1}, $\pi_1$ is equitable. Fig. \ref{fig1-b} shows the corresponding quotient digraph ${\mathcal G}/\pi_1$, where the vertex corresponding to the only one non-singleton cell $\{v_1,v_2\}$ is denoted as $v_{1,2}$, and the weight of edge $(v_{1,2}, v_{1,2})$ is $A({\mathcal G})_{12}=1$ because of $v_2\in C_1$ and $N_{out}(v_2)\cap C_1=\{v_1\}$. The adjacency matrix of ${\mathcal G}/\pi_1$ is 
\begin{equation*}
    H=\left[\begin{array}{ccc}
1 & 1 & 0\\
0 & 0 & 1\\
0 & 0 & 0
\end{array}\right],
\end{equation*}
which satisfies $P^{\top}A=HP^{\top}.$

For the equitable partitions $\pi_2=\{\{v_1,v_2,v_3\},\{v_4\}\}$ and $\pi_3=\{\{v_1,v_3\},\{v_2\},\{v_4\}\}$, their respective quotient digraphs are shown in Fig. \ref{fig1-c} and Fig. \ref{fig1-d}. In these figures, vertex $v_{1,2,3}$ corresponds to the non-singleton cell $\{v_1,v_2,v_3\}$ and vertex $v_{1,3}$ corresponds to $\{v_1,v_3\}$.

Since $\{v_1,v_2,v_3\}=\{v_1,v_2\}\cup\{v_3\}=\{v_1,v_3\}\cup\{v_2\}$, $\pi_1$ and $\pi_3$ are finer than $\pi_2$. The join of $\pi_1$ and $\pi_3$ is $\{\{v_1\},\{v_2\},\{v_3\},\{v_4\}\}$ and their meet is $\pi_2$. That is
$$\pi_2\preceq\pi_1,\quad\pi_2\preceq\pi_3,$$
$$\pi_1\sqcup\pi_3=\{\{v_1\},\{v_2\},\{v_3\},\{v_4\}\},$$
$$\pi_1\sqcap\pi_3=\pi_2.$$
\begin{figure}
\centering
\setlength{\unitlength}{10mm}
\subcaptionbox{A digraph ${\mathcal G}$.\label{fig1-a}}{
\begin{tikzpicture}[->,thick]
	\tikzstyle{every node}=[draw,circle,radius=1mm,inner sep=0pt,minimum size=0.8em];
	\node (v1)[label=below:{$v_1$}]at (3,0){};
	\node (v2)[label=above:{$v_2$}]at (2,0.8){};
	\node (v3)[label=below:{$v_3$}]at (2,0){};
	\node (v4)[label=below:{$v_4$}]at (1,0){};
	
	\path (v2) edge (v1);
	\path (v3) edge (v1);
	\path (v4) edge (v3);
	\draw (v1) to [out=45,in=315,looseness=10] (v1);
	\end{tikzpicture}
	}
\quad\quad\quad\quad
\subcaptionbox{${\mathcal G}/\pi_1$, where $\pi_1=\{\{v_1,v_2\},\{v_3\},\{v_4\}\}$.\label{fig1-b}}{
	\begin{tikzpicture}[->,thick]
		\tikzstyle{every node}=[draw,circle,radius=1mm,inner sep=0pt,minimum size=0.8em];
		\node (v1)[label={[shift={(0.0,-0.0)}]$v_{1,2}$}]at (3,0){};
		\node (v2)[label={[shift={(0.0,-0.8)}]$v_3$}]at (2,0){};
		\node (v3)[label={[shift={(0,-0.8)}]$v_4$}]at (1,0){};
		
		\path (v2) edge (v1);
		\path (v3) edge (v2);
		\draw (v1) to [out=45,in=315,looseness=10] (v1);
		
		\end{tikzpicture}
		}

	\subcaptionbox{${\mathcal G}/\pi_2$, where $\pi_2=\{\{v_1,v_2,v_3\},\{v_4\}\}$.\label{fig1-c}}{
		\begin{tikzpicture}[->,thick]
			\tikzstyle{every node}=[draw,circle,radius=1mm,inner sep=0pt,minimum size=0.8em];
			\node (v1)[label={[shift={(0.0,-0.2)}]$v_{1,2,3}$}]at (3,0){};
			\node (v2)[label={[shift={(0.0,-0.8)}]$v_4$}]at (1.0,0){};
			
			\path (v2) edge (v1);
			\draw (v1) to [out=45,in=315,looseness=10] (v1);
			\end{tikzpicture}
			}
		 \quad\quad\quad\quad
		\subcaptionbox{${\mathcal G}/\pi_3$, where $\pi_3=\{\{v_1,v_3\},\{v_2\},\{v_4\}\}$.\label{fig1-d}}{
			\begin{tikzpicture}[->,thick]
				\tikzstyle{every node}=[draw,circle,radius=1mm,inner sep=0pt,minimum size=0.8em];
				\node (v1)[label={[shift={(0.2,-0.0)}]$v_{1,3}$}]at (3,0){};
				\node (v2)[label=left:{$v_2$}]at (1.5,0.5){};
				\node (v3)[label=left:{$v_4$}]at (1.5,-0.5){};
				
				\path (v2) edge (v1);
				\path (v3) edge (v1);
				\draw (v1) to [out=45,in=315,looseness=10] (v1);
				\end{tikzpicture}
			}	
			\caption{A digraph ${\mathcal G}$ and its three quotient digraphs.}
			\label{fig2.1}
			\end{figure}
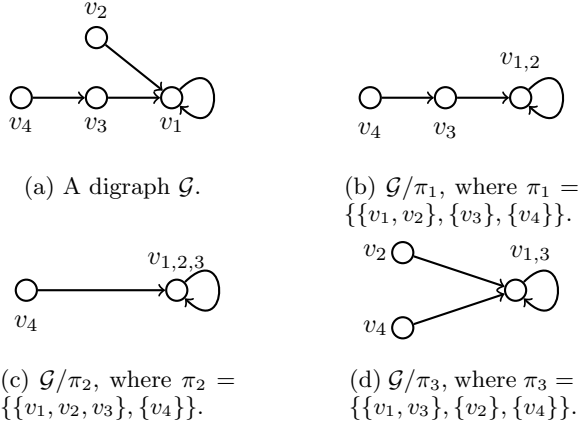
			\end{exa}
			
			\subsection{The semi-tensor product (STP) of matrices}
Notation:
			
			\begin{itemize}
				\item ${\mathcal M}_{n \times m}$: the set of $n\times m$ real matrices.
				\item $\delta_{n}^{i}$: the $i$th column of the identity matrix $I_{n}$.
				\item $\Delta_{n}=\left\{\delta_{n}^{i} \mid i=1,2, \ldots, n\right\}$.
				\item $\operatorname{Col}(A)$: the set of columns of $A$.
				\item $ {\mathcal R}(A)$: the row space of $A$.
				\item A matrix $L \in {\mathcal M}_{n \times m}$ is called a logical matrix if
				$\operatorname{Col}(L) \subseteq \Delta_{n}$. Denote the set of $n \times m$ logical matrices by $\mathcal{L}_{n \times m}$.
				\item Matrix $\left[\delta_{n}^{i_{1}}, \ldots, \delta_{n}^{i_{m}}\right] \in \mathcal{L}_{n \times m}$ can be expressed as $\delta_{n}\left[i_{1}, \ldots, i_{m}\right]$ for brevity.
				\item $\sim$: one-to-one correspondence between binary logical values in ${\mathcal D}:=\{0,1\}$ and vectors in $\Delta_{2}$. That is, $ 1 \sim \delta_{2}^{1}$ and $0 \sim \delta_{2}^{2}$.
				\item $[n;m]$: the set of integers $x$ with $n\leq x\leq m$.
			\end{itemize}
			
			\begin{dfn}\citep{che12,che09}\label{d2.1.1}
				Let $M\in {\mathcal M}_{m\times n}$, $N\in {\mathcal M}_{p\times q}$, and $t=\lcm\{n,p\}$ be the least common multiple of $n$ and $p$.
				The semi-tensor product (STP) of $M$ and $N$, denoted by $M\ltimes N$, is defined as
				\begin{align}\label{2.1.1}
					\left(M\otimes I_{t/n}\right)\left(N\otimes I_{t/p}\right)\in {\mathcal M}_{mt/n\times qt/p},
				\end{align}
				where $\otimes$ is the Kronecker product.
			\end{dfn}
			
			
\begin{dfn}\citep{che09}
Let $A\in \mathcal{M}_{p\times n}$ and $B\in \mathcal{M}_{q\times n}$. The Khatri-Rao Product of $A$ and $B$ is defined as follows.
\begin{align}
\begin{array}{ll}
A*B\!\!&=\![\Col_1(\!A\!)\!\ltimes\!\Col_1(\!B\!)\!,\ldots,\!\Col_n(A)\!\ltimes\!\Col_n(B)]\\
&\in \mathcal{M}_{pq\times n}.
\end{array}
\end{align}
\end{dfn}
			
\subsection{\texorpdfstring{$M$}--invariant dual subspaces of BNs}
A BN can be expressed as a set of Boolean functions
\begin{align}\label{2.2.1}
	x_i(t+1)=f_i(x_1(t),\ldots,x_n(t)), i\in[1;n],
\end{align}
where $x_i(t)\in {\mathcal D}$ and $f_i:{\mathcal D}^n\rightarrow {\mathcal D}$, $i\in[1;n]$. Equivalently, a BN is represented as 
\begin{align}\label{2.2.3}
    x(t+1) = {\mathsf F}(x(t)),
\end{align}
where $x(t)=[x_1(t),...,x_n(t)]^{\top}\in{\mathcal D}^n$.
We identify $x_i\in {\mathcal D}$ with vector $\bar{x}_i\in\Delta_2$, $i\in[1;n]$, denoted as $x_i\sim\bar{x}_i$, by the one-to-one correspondence $ 1 \sim \delta_{2}^{1}$ and $0 \sim \delta_{2}^{2}$. 
Define $\bar{x}:=\ltimes_{i=1}^n\bar{x}_i$. We extend the relation $\sim$ as follows:
$$(x_1,x_2,\ldots,x_n):=x\sim\bar{x},$$
where $x\in{\mathcal D}^n$ and $\bar{x}\in\Delta_{2^n}$.
Furthermore, we note 
$${\mathcal D}^n\sim\Delta_{2^n}.$$
Thus, a Boolean function can be expressed in the following form
\begin{align}\label{2.2.2}
	f_i(x_1,x_2,\ldots,x_n)\sim M_i\bar{x},~i\in[1;n].
\end{align}
where $M_i\in {\mathcal L}_{2\times 2^n}$ is the structure matrix of $f_i$. Moreover, the algebraic state space representation (ASSR) of BN $(\ref{2.2.1})$ is as follows,
\begin{align}\label{2.2.4}
\bar{x}(t+1)=M\bar{x}(t),
\end{align}
where
$
M=M_1*M_2*\cdots*M_n\in {\mathcal L}_{2^n\times 2^n}
$
is called the transition matrix of $(\ref{2.2.1})$.	

For BN $(\ref{2.2.1})$, its state space is defined as ${\mathcal X}:={\mathcal D}^n$. The dual (state) space is defined as the set of all Boolean functions of the state variables $x_1, x_2, \ldots, x_n$. We denote the dual space as ${\mathcal X}^{*}:={\mathcal F}_{\ell}\{x_1, x_2, \ldots, x_n\}$ \citep{che10,zhang24}.
For $z_1, z_{2}, \ldots, z_r\in {\mathcal X}^{*}$, the dual subspace generated by $z_1, z_2, \ldots, z_r$ is defined as
${\mathcal Z}^{*}:={\mathcal F}_{\ell}\{z_1, z_2, \ldots, z_r\}$, which is the set of all Boolean functions of Boolean functions $z_1, z_2, \ldots, z_r$.
%
Let $\bar{z}=\ltimes_{i=1}^r\bar{z}_i$. Then $\bar{z}=G\bar{x},$
where $G\in {\mathcal L}_{2^r\times 2^n}$ is called the structure matrix of ${\mathcal Z}^{*}$ \citep{che10}.

For any Boolean function $f(z(x))=(f\circ z)(x)\in {\mathcal Z}^{*}={\mathcal F}_{\ell}\{z_1, z_2, \ldots, z_r\}=:{\mathcal F}_{\ell}\{z\}$, its logical form satisfies 
$$f(z(x))\sim F{\bar z}=FG{\bar x},$$
where $F\in {\mathcal L}_{2\times 2^r}$ is the structure matrix of $f$. 
Thus, ${\mathcal Z}^{*}$ can be identified as the set of Boolean functions with structure matrices in $\{FG\in {\mathcal L}_{2\times 2^n}|F\in {\mathcal L}_{2\times 2^r}\}$.


The following definition is an equivalent logical form of the $M$-invariant dual subspace as presented in \citep{che21}, where $M$ is the structure matrix of $\mathsf F$.


\begin{dfn}
\label{d2.7}
Given a dual subspace ${\mathcal Z}^{*}={\mathcal F}_{\ell}\{z\}:={\mathcal F}_{\ell}\{z_1, z_2, \ldots, z_r\}\subset{\mathcal X}^{*}$, ${\mathcal Z}^{*}$ is called \emph{$M$-invariant} if for every Boolean function $f(z)=f(z(x))=(f\circ z)(x)\in{\mathcal Z}^{*}$, the Boolean function $(f\circ z)({\mathsf F}(x))$ resulted from the action of the BN dynamics {\green $\mathsf F$} on $x$ still belongs to ${\mathcal Z}^{*}$.
\end{dfn}

\begin{remr}
\label{explanition of M}
In the algebraic form of Boolean functions, ${\mathcal Z}^{*}$ is called \emph{$M$-invariant} if its structure matrix $G$ satisfies the following condition: for every $F\in {\mathcal L}_{2\times 2^r}$ there exists $F'\in {\mathcal L}_{2\times 2^r}$ such that $FGM=F'G$ \citep{che21}.
\end{remr}

\begin{them}\citep{che21}\label{l2.8}
Consider BN $(\ref{2.2.1})$ with its ASSR $(\ref{2.2.4})$.
A dual subspace ${\mathcal Z}^{*}$ is $M$-invariant with respect to $(\ref{2.2.1})$ if there exists a logical matrix $H\in {\mathcal L}_{2^r\times 2^r}$ such that
\begin{align} \label{3.1.101}
GM=HG.
\end{align}
\end{them}

\begin{remr}
In light of Theorem \ref{l2.8}, the matrix $F'$ in Remark \ref{explanition of M} can be identified as $FH$.
\end{remr}

\begin{lemm}\citep{che21}\label{t2.8}
A dual subspace ${\mathcal Z}^{*}$ is $M$-invariant iff its dynamics can be expressed as
\begin{align} \label{3.1.102}
\bar{z}(t+1)=H\bar{z}(t),
\end{align}
where $H\in {\mathcal L}_{2^r\times 2^r}$.
\end{lemm}

The dynamics (\ref{3.1.102}) is called the dual dynamics of BN (\ref{2.2.1}) with respect to ${\mathcal Z}^{*}$ and $H$ is the state dual transition matrix.




\begin{lemm}\citep{che21} \label{t2.10}
Assume that ${\mathcal Z}^{*}_i$, $i=1,2$, are $M$-invariant dual subspaces. 
Then
$
{\mathcal Z}^{*}={\mathcal Z}^{*}_1\cup {\mathcal Z}^{*}_2
$
is also 
$M$-invariant. 
\end{lemm}

\begin{remr}
\begin{itemize}
    \item Given BN (\ref{2.2.1}). According to \citep{bche11}, ${\mathcal Z}^{*}_1={\mathcal F}_{\ell}\{z^1_1,z^1_2,\ldots,z_r^1\}$ is called a regular dual subspace if there exist $(z^2_1,z^2_2,\ldots,z_{n-r}^2)\subset {\mathcal X}^{*}$ such that $(z^1_1,z^1_2,\ldots,z_r^1,z^2_1,z^2_2,\ldots,z_{n-r}^2)$ is another coordinate frame. Moreover, ${\mathcal Z}^{*}_1$ is regular iff its structure matrix $G\in {\mathcal L}_{2^r\times 2^n}$ satisfies
\begin{equation}
    \label{regular}
    \sum_{i=1}^{2^n}[G]_{ji}=2^{n-r},~j=1,2,\ldots,2^{r}.
\end{equation}
\item According to the definition in \citep{bche11}, an invariant dual subspace was defined as a dual subspace that is both regular and $M$-invariant. Since regularity is a strong requirement, $M$-invariant dual subspace is a more general concept than invariant dual subspace.
\end{itemize}
\end{remr}

\begin{exa}{\label{e3.2}}
A Boolean equation about the gene network of the $\lambda$ bacteriophage can be expressed in the following form
\begin{equation}{\label{ep-new2-3.8}}
\left\{
\begin{aligned}
N(t+1) &=[\neg cI(t)] \wedge[\neg cro(t)], \\
cI(t+1) &=[\neg cro(t)] \wedge[cI(t) \vee cII(t)], \\
cII(t+1) &=[\neg cI(t)]  \wedge[N(t) \vee cIII(t)], \\
cIII(t+1) &=[\neg cI(t)] \wedge N(t), \\
cro(t+1) &=[\neg cI(t)] \wedge[\neg cII(t)].
\end{aligned}
\right.
\end{equation}
where $N(t),cI(t),cII(t),cIII(t),cro(t)\in {\mathcal D}$. Suppose 
$(x_1,x_2,x_3,x_4,x_5)=(N,cI,cII,cIII,cro)$.
Let $\bar{x}(t):=\ltimes_{i=1}^5\bar{x}_i(t)$, where ${x}_i\sim\bar{x}_i\in \Delta_2$. The algebraic form of every Boolean function $f_i$ is
 \begin{equation}
			\bar{x}_i(t+1) =M_i\bar{x}(t),~i\in\left[1;5\right],
\end{equation}
where 
\begin{equation*}
\begin{array}{l l l l l l l l l l l l l l l l l l}
M_1 =\delta_{2}[&2 &2 &2 &2 &2 &2 &2 &2 &2 &1 &2 &1 &2 &1 &2 &1\\
&2 &2 &2 &2 &2 &2 &2 &2 &2 &1 &2 &1 &2 &1 &2 &1],\\
M_2 =\delta_{2}[&2 &1 &2 &1 &2 &1 &2 &1 &2 &1 &2 &1 &2 &2 &2 &2\\
&2 &1 &2 &1 &2 &1 &2 &1 &2 &1 &2 &1 &2 &2 &2 &2],\\
M_3 =\delta_{2}[&2 &2 &2 &2 &2 &2 &2 &2 &1 &1 &1 &1 &1 &1 &1 &1\\
&2 &2 &2 &2 &2 &2 &2 &2 &1 &1 &2 &2 &1 &1 &2 &2 ],\\
M_4 =\delta_{2}[&2 &2 &2 &2 &2 &2 &2 &2 &1 &1 &1 &1 &1 &1 &1 &1\\
&2 &2 &2 &2 &2 &2 &2 &2 &2 &2 &2 &2 &2 &2 &2 &2 ],\\
M_5 =\delta_{2}[&2 &2 &2 &2 &2 &2 &2 &2 &2 &2 &2 &2 &1 &1 &1 &1\\
&2 &2 &2 &2 &2 &2 &2 &2 &2 &2 &2 &2 &1 &1 &1 &1 ].\\
\end{array} 
\end{equation*}
Moreover, the ASSR of (\ref{ep-new2-3.8}) is 
$${\bar x}(t+1)=M{\bar x}(t),$$ where
$$
\begin{array}{ll}
M&=M_1*M_2*M_3*M_4*M_5\\
&=\begin{array}{l l l l l l l l l l l l l l l l l l}
\delta_{32}[\!\!&32&24&32&24&32&24&32&24&26&2&26&2&25&9&25&9\\
&32&24&32&24&32&24&32&24&28&4&32&8&27&11&31&15].
\end{array}
\end{array}
$$
Consider $z_1,z_2\in {\mathcal X}^{*}$, where
$$\bar{z}_1=G_1\bar{x}, \bar{z}_2=G_2\bar{x}$$ and
$$
\begin{array}{l l l l l l l l l l l l l l l l l l}
G_1=\delta_{2}[&2&1&2&1&2&1&2&1&1&1&1&1&1&1&1&1\\
&2&1&2&1&2&1&2&1&1&1&2&1&2&1&1&1],\\
G_2=\delta_{2}[&1&1&1&1&1&1&1&1&1&1&1&1&1&1&1&1\\
&1&1&1&1&1&1&1&1&1&1&1&1&2&1&1&2].
\end{array}
$$
The structure matrix of dual subspace ${\mathcal Z}^{*}:={\mathcal F}_{\ell}\{z_1,z_2\}$ generated by $z_1$ and $z_2$ is
\begin{equation}
\begin{array}{l l l l l l l l l l l l l l l l l l}
G&=G_1*G_2\\
&=\begin{array}{l l l l l l l l l l l l l l l l l l}
\delta_{4}[&3&1&3&1&3&1&3&1&1&1&1&1&1&1&1&1\\
&3&1&3&1&3&1&3&1&1&1&3&1&4&1&1&2].
\end{array}
\end{array}
\label{e-13}
\end{equation}
There exists an $H=\begin{array}{cc c c c}
\delta_{4}[1 &1 &2 &3]
\end{array}$ satisfying $GM=HG$. According to Lemma \ref{t2.8}, we get that ${\mathcal Z}^{*}$ is $M$-invariant. 

Under the original BN (\ref{ep-new2-3.8}), the dual dynamics with respect to ${\mathcal Z}^{*}$ is
\begin{align}
\begin{array}{l l }
\bar{z}(t+1)&=G\bar{x}(t+1)=GM\bar{x}(t)=HG\bar{x}(t)\\
&=H\bar{z}(t),
\end{array}
\end{align}
where $\bar{z}=\bar{z}_1\ltimes\bar{z}_2$.
Since $$\begin{aligned}
H=&H_1*H_2\\
=&\begin{array}{cc c c c}
\delta_{2}[&1 &1 &1 &2]
\end{array}*\begin{array}{cc c c c}
\delta_{2}[&1 &1 &2 &1],
\end{array}
\end{aligned}$$
the dynamics of ${\mathcal Z}^{*}$ can be expressed in logical forms as 
\begin{equation}
\left\{
\begin{aligned}
    z_1(t+1) &=h_1(z(t)), \\
    z_2(t+1) &=h_2(z(t)),
\end{aligned}
\right.
\end{equation}
where the structure matrices of $h_1, h_2\in{\mathcal Z}^{*}$ are $H_1$ and $H_2$, respectively. 

Since ${\mathcal Z}^{*}$ is $M$-invariant and not regular based on (\ref{regular}), ${\mathcal Z}^{*}$ is not an invariant dual subspace.
\end{exa}

\section{A graph representation of a BN and its \texorpdfstring{$M$}--invariant dual subspaces}\label{III}

In this section, we establish a bijection between dual subspaces of a BN and partitions of its state set. We prove that two lattices defined on the set of dual subspaces and all partitions are isomorphic. Moreover, we reveal that a dual subspace is $M$-invariant iff the corresponding partition is equitable. Based on these results, we give a complete structural characterisation of the smallest $M$-invariant dual subspaces generated by a set of Boolean functions.

\subsection{Dual subspaces and partitions}\label{III-A}
Given a dual subspace ${\mathcal Z}^{*}={\mathcal F}_{\ell}\{z_1,\ldots, z_r\}$, recall that $\bar{z}:=\ltimes_{i=1}^r\bar{z}_i=G\bar{x}$, where $G\in {\mathcal L}_{2^r\times 2^n}$ is the structure matrix. 
We define a partition ${\pi_G}:=\{\{\bar{x}|G\bar{x}=\delta_{2^r}^{i}\}| i\in[1;2^r],\text{ there is at least one }x\in\Delta_{2^n}\text{ such that }Gx=\delta_{2^r}^i\}$ of state set $\Delta_{2^n}$.
We observe that $G^{\top}$ is exactly the characteristic matrix of ${\pi_G}$. Furthermore, it is evident that ${\pi_G}$ is uniquely determined by the row space ${\mathcal R}(G)$ of $G$.
That is,
${\pi_{G_1}}={\pi_{G_2}}$ iff ${\mathcal R}(G_1)={\mathcal R}({G_2})$. 

As a matter of fact, ${\mathcal Z}^{*}$ is also uniquely determined by ${\mathcal R}(G)$. As mentioned before, ${\mathcal Z}^{*}$ can be identified as a collection of Boolean functions with structure matrices $\{FG\in {\mathcal L}_{2\times 2^n}|F\in {\mathcal L}_{2\times 2^r}\}$. Since $\{FG_1\in {\mathcal L}_{2\times 2^n}|F\in {\mathcal L}_{2\times 2^r}\}=\{FG_2\in {\mathcal L}_{2\times 2^n}|F\in {\mathcal L}_{2\times 2^r}\}$ iff ${\mathcal R}(G_1)={\mathcal R}({G_2})$, two dual subspaces are equal iff their structure matrices $G_1$ and $G_2$ satisfy ${\mathcal R}(G_1)={\mathcal R}({G_2})$.

Next, we establish a bijection between dual subspaces and partitions of the state set.


\begin{them} 
\label{l3.1}
Let ${\mathcal X}^{**}$ be the family of dual subspaces over ${\mathcal X}\sim\Delta_{2^n}$, and let $\Pi$ be the family of partitions of $\Delta_{2^n}$. Define a mapping ${\mathcal P}:{\mathcal X}^{**}\rightarrow \Pi$ by ${\mathcal P}({\mathcal Z}^{*}):={\pi_G}$, where $G$ is the structure matrix of the dual subspace ${\mathcal Z}^{*}$. Then, ${\mathcal P}$ is a bijection.
\end{them}

\begin{proof}
$\mathcal{P}$ is indeed a mapping because as shown before,  $\mathcal{Z}_1^*=\mathcal{Z}_2^*$ implies $\pi_{G_1}=\pi_{G_2}$, where $G_1$ and $G_2$ are the structure matrices of $\mathcal{Z}_1^*$ and $\mathcal{Z}_2^*$, respectively.

				
	
	(1) ${\mathcal P}$ is surjective.
 For any partition $\pi\in \Pi$ with the characteristic matrix $G^{\top}$, we can construct a dual subspace whose structure matrix has row space ${\mathcal R}(G)$. (Note that the row number of a structure matrix is always a power of 2. In some cases, we may need to add all-zero rows to satisfy this condition.)
	
	(2) ${\mathcal P}$ is injective. Consider two dual subspaces ${\mathcal Z}^{*}_1$ and ${\mathcal Z}^{*}_2$ with structure matrices $G_1$ and $G_2$, respectively. If ${\mathcal P}({\mathcal Z}^{*}_1)={\mathcal P}({\mathcal Z}^{*}_2)$ (i.e., $\pi_{G_1}=\pi_{G_2}$), then ${\mathcal R}(G_1)={\mathcal R}({G_2})$. It follows that ${\mathcal Z}^{*}_1= {\mathcal Z}^{*}_2$.
		
	Therefore, ${\mathcal P}$ is a bijection.
\end{proof}


\begin{lemm}
	\label{l3.2}
	Given ${\mathcal Z}^{*}_{1},{\mathcal Z}^{*}_{2} \in {\mathcal X}^{**}$, ${\mathcal P}({\mathcal Z}^{*}_1)\preceq{\mathcal P}({\mathcal Z}^{*}_2)$ iff ${\mathcal Z}^{*}_{1}\subseteq{\mathcal Z}^{*}_{2}$.
\end{lemm}
\begin{proof}
	Let $G_1$ and $G_2$ be the structure matrices of ${\mathcal Z}^{*}_{1}$ and ${\mathcal Z}^{*}_{2}$, respectively. Then, ${\pi_{G_1}}\preceq{\pi_{G_2}}$ iff ${\mathcal R}(G_1)\subseteq{\mathcal R}(G_2)$. Moreover, ${\mathcal Z}^{*}_{1}\subseteq{\mathcal Z}^{*}_{2}$ iff ${\mathcal R}(G_1)\subseteq{\mathcal R}(G_2)$. Thus, ${\mathcal P}({\mathcal Z}^{*}_1)\preceq{\mathcal P}({\mathcal Z}^{*}_2)$ iff ${\mathcal Z}^{*}_{1}\subseteq{\mathcal Z}^{*}_{2}$.
\end{proof}

Since ${\mathcal P}$ is a bijection, ${\mathcal P}^{-1}$ exists. From Lemma \ref{l3.2}, ${\mathcal P}$ and ${\mathcal P}^{-1}$ are both \em{order-preserving}. Thus two lattices $({\mathcal X}^{**},\subseteq)$ and $(\Pi,\preceq)$ are \em{isomorphic} according to \citep[Theorem 14.2]{che12}. It is straightforward to show the following proposition.

\begin{prp}
	\label{l3.3}
	Consider ${\mathcal Z}^{*}_{1},{\mathcal Z}^{*}_{2}\in {\mathcal X}^{**}$.
	\begin{itemize}
		\item[(i)] ${\mathcal P}({\mathcal Z}^{*}_{1}\cap{\mathcal Z}^{*}_{2})={\mathcal P}({\mathcal Z}^{*}_1)\sqcap {\mathcal P}({\mathcal Z}^{*}_2)$.
		\item[(ii)] ${\mathcal P}({\mathcal Z}^{*}_{1}\cup{\mathcal Z}^{*}_{2})={\mathcal P}({\mathcal Z}^{*}_1)\sqcup {\mathcal P}({\mathcal Z}^{*}_2)$.
	\end{itemize}
\end{prp}

\begin{remr}\label{p3.3}
	Considering various dual subspaces, we conclude the following properties of their corresponding partitions.	
	
	\begin{itemize}
		\item [(i)] Given that ${\mathcal Z}^{*}\subseteq{\mathcal X}^{*}$, it follows that ${\mathcal P}({\mathcal Z}^{*})\preceq {\mathcal P}({\mathcal X}^{*})$. Moreover, ${\mathcal P}({\mathcal X}^{*})=\{\{\delta_{2^n}^1\},\{\delta_{2^n}^2\},\ldots,\{\delta_{2^n}^{2^n}\}\}$ constitutes the finest partition of the state set.

		\item [(ii)] Consider the case where ${\mathcal Z}^{*}={\mathcal F}_{\ell}\{z\}$ for $z\in{\mathcal X}^{*}$. Partition ${\mathcal P}({\mathcal Z}^{*})$ is a 2-partition with cells $\{C, \Delta_{2^n}\backslash C\}$, where $C:=\{{\bar x}\in \Delta_{2^n}|z(x)\sim\delta_2^1\}$.
  For notational brevity, we denote by ${\mathcal Z}^{*}_C$ the dual subspace whose corresponding partition is $\{C, \Delta_{2^n}\backslash C\}$. Moreover, we express ${\mathcal P}({\mathcal Z}^{*})$ as ${\mathcal P}(z)$.
		
		\item [(iii)] Given ${\mathcal Z}^{*}$, ${\mathcal P}(z)\preceq {\mathcal P}({\mathcal Z}^{*})$ for all $z\in{\mathcal Z}^{*}$.
		\item [(iv)] For ${\mathcal Z}^{*}={\mathcal F}_{\ell}\{z_1, z_2, \ldots, z_r\}$, ${\mathcal P}({\mathcal Z}^{*})={{\mathcal P}(z_1)}\sqcup{{\mathcal P}(z_2)}\sqcup\cdots\sqcup{{\mathcal P}(z_r)}$.
	\end{itemize}
\end{remr}

\subsection{\texorpdfstring{$M$}--invariant dual subspaces and equitable partitions}\label{III-AA}

The partitions corresponding to $M$-invariant dual subspaces possess the following properties.


\begin{them}{\label{t3.3}}
	For BN $(\ref{2.2.1})$ with STG ${\mathcal G}$, a dual subspace ${\mathcal Z}^{*}$ is $M$-invariant with respect to $(\ref{2.2.1})$ iff ${\mathcal P}({\mathcal Z}^{*})$ is equitable. Moreover, the dual transition matrix $H$ of ${\mathcal Z}^{*}$ is precisely the adjacency matrix of the quotient digraph ${\mathcal G}/{\mathcal P}({\mathcal Z}^{*})$.
\end{them}
\begin{proof}
	Suppose that a dual subspace ${\mathcal Z}^{*}$ with structure matrix $G$ is $M$-invariant. Then $G^{\top}$ is the characteristic matrix of ${\mathcal P}({\mathcal Z}^{*})$ according to Theorem \ref{l3.1}. By Theorem \ref{l2.8}, there exists a logical matrix $H$ such that 
	\begin{equation}
		\label{e12}
		GM=HG.
	\end{equation}
	For ${\mathcal G}$, the state set $\Delta_{2^n}$ is its vertex set, and $M$ is its adjacency matrix. By Lemma \ref{l2.2}, (\ref{e12}) implies that 
	${\mathcal P}({\mathcal Z}^{*})$ is equitable and $H$ is the adjacency matrix of ${\mathcal G}/{\mathcal P}({\mathcal Z}^{*})$. 
	%
\end{proof}
\begin{exa}
Let's consider the BN (\ref{ep-new2-3.8}) in Example \ref{e3.2}. Its STG ${\mathcal G}$ is illustrated in Fig. \ref{stg of exm2}. For the dual subspace ${\mathcal Z}^{*}$ with structure matrix $G$ as given in (\ref{e-13}), the characteristic matrix of ${\mathcal P}({\mathcal Z}^{*})$ is $G^{\top}$. In Fig. \ref{stg of exm2}, vertices are color-coded to represent the distinct cells of ${\mathcal P}({\mathcal Z}^{*})$. 

In Example \ref{e3.2}, we have previously proven that ${\mathcal Z}^{*}$ is $M$-invariant. According to Theorem \ref{t3.3}, we can deduce that ${\mathcal P}({\mathcal Z}^{*})$ is equitable. The corresponding quotient digraph ${\mathcal G}/{\mathcal P}({\mathcal Z}^{*})$ is depicted in Fig \ref{fig4-b}. The adjacency matrix of ${\mathcal G}/{\mathcal P}({\mathcal Z}^{*})$ is precisely $H$, the dual state transition matrix of ${\mathcal Z}^{*}$.

\begin{figure}
\centering
\subcaptionbox{${\mathcal G}$, the STG of BN (\ref{ep-new2-3.8}).\label{stg of exm2}}
{\begin{tikzpicture}[->,thick]
\tikzstyle{every node}=[draw,circle,radius=0.8mm,inner sep=0pt,minimum size=1.5em];
\node (v1) [fill=d2] at (1,0.8){$1$};
\node (v2)[fill=d4] at (6,2.4){$2$};
\node (v3)[fill=d2] at (1,1.6){$3$};
\node (v4)[fill=d4] at (6,1.6){$4$};
\node (v5)[fill=d2] at (1,2.4){$5$};
\node (v6)[fill=d4] at (6,0.8){$6$};
\node (v7)[fill=d2] at (1,-0.8){$7$};
\node (v8)[fill=d4] at (6,0){$8$};
\node (v9)[fill=d4] at (4,2.4){$9$};
\node (v10)[fill=d4] at (5,3.2){$10$};
\node (v11)[fill=d4]  at (4,1.6){$11$};
\node (v12)[fill=d4]  at (5,2.4){$12$};
\node (v13)[fill=d4]  at (3,-0.8){$13$};
\node (v14)[fill=d4] at (3,2.4){$14$};
\node (v15)[fill=d4]  at (3,0){$15$};
\node (v16)[fill=d4]  at (3,3.2){$16$};
\node (v17)[fill=d2]  at (1,-1.6){$17$};
\node (v18)[fill=d4]  at (6,-0.8){$18$};
\node (v19)[fill=d2]  at (1,-2.4){$19$};
\node (v20)[fill=d4]  at (6,-1.6){$20$};
\node (v21)[fill=d2]  at (1,-3.2){$21$};
\node (v22)[fill=d4]  at (6,-2.4){$22$};
\node (v23)[fill=d2]  at (1,3.2){$23$};
\node (v24)[fill=d4]  at (7,0.4){$24$};
\node (v25)[fill=d4]  at (4,0){$25$};
\node (v26)[fill=d4]  at (5,1.6){$26$};
\node (v27)[fill=d2]  at (1,0){$27$};
\node (v28)[fill=d4]  at (5,0){$28$};
\node (v29)[fill=d1]  at (0,0){$29$};
\node (v30)[fill=d4]  at (3,1.6){$30$};
\node (v31)[fill=d4]  at (7,-1){$31$};
\node (v32)[fill=c2]  at (2,0){$32$};

\path (v29) edge (v27);
\path (v1) edge (v32);
\path (v3) edge (v32);
\path (v5) edge (v32);
\path (v7) edge (v32);
\path (v17) edge (v32);
\path (v19) edge (v32);
\path (v21) edge (v32);
\path (v23) edge (v32);
\path (v27) edge (v32);
\path (v32) edge (v15);
\path (v14) edge (v9);
\path (v16) edge (v9);
\path (v30) edge (v11);
\path (v13) edge (v25);
\path (v15) edge (v25);
\path (v9) edge (v26);
\path (v11) edge (v26);
\path (v25) edge (v28);
\path (v10) edge (v2);
\path (v12) edge (v2);
\path (v26) edge (v4);
\path (v28) edge (v8);
\path (v2) edge (v24);
\path (v4) edge (v24);
\path (v6) edge (v24);
\path (v8) edge (v24);
\path (v18) edge (v24);
\path (v20) edge (v24);
\path (v22) edge (v24);
\draw (v24) to [out=45,in=315,looseness=6] (v24);
\draw (v31) to [out=45,in=315,looseness=6] (v31);




\end{tikzpicture}
}

\subcaptionbox{The quotient digraph of ${\mathcal P}({\mathcal Z}^{*})$, where $G^{\top}$ is its characteristic matrix.\label{fig4-b}}
{\begin{tikzpicture}[->,thick]
				\tikzstyle{every node}=[draw,circle,radius=0.2mm,inner sep=0pt,minimum size=1em];
				\def \n {8}
				\def \radius {0.7cm}
				\def \margin {25}
				\node (v1) [fill=d4]at (3, 1-1){};
				\put(100,-5){$e$}
				\node (v5)[fill=c2] at (2, 1-1){};
				\node (v7)[fill=d2] at (1, 1-1){};
				\node (v8)[fill=d1] at (-0.0,1-1){};
				\path (v8) edge (v7);
				\path (v7) edge (v5);
				\path (v5) edge (v1);
				\draw (v1.-60) arc (210:180+300:2mm);
			\end{tikzpicture}
	}
\caption{An illustration of $M$-invariant dual subspace using partitions.}
\label{new-fig3.1}
\end{figure}
\end{exa}

For BN (\ref{2.2.1}) and a given dual subspace ${\mathcal Z}^{*}$, Cheng et al. \citep{che21} gave an algorithm to compute the smallest $M$-invariant dual subspace $\overline{{\mathcal Z}^{*}}$ containing ${\mathcal Z}^{*}$. We call $\overline{{\mathcal Z}^{*}}$ {\em the smallest $M$-invariant dual subspace generated by} ${\mathcal Z}^{*}$.
Based on Lemma \ref{l3.2}, we get ${\mathcal P}({\mathcal Z}^{*})\preceq{\mathcal P}(\overline{{\mathcal Z}^{*}})$ from ${\mathcal Z}^{*}\subseteq\overline{{\mathcal Z}^{*}}$.
On the other hand, for any $M$-invariant dual subspace $\overline{{\mathcal Z}^{*}_1}$ containing ${\mathcal Z}^{*}$, it follows from the term ``smallest'' that $\overline{{\mathcal Z}^{*}}\subseteq \overline{{\mathcal Z}^{*}_1}$. Thus, ${\mathcal P}(\overline{{\mathcal Z}^{*}})\preceq {\mathcal P}(\overline{{\mathcal Z}^{*}_1})$ by Lemma \ref{l3.2}. According to Theorem \ref{t3.3}, ${\mathcal P}(\overline{{\mathcal Z}^{*}})$ and ${\mathcal P}(\overline{{\mathcal Z}^{*}_1})$ are both equitable. We conclude that ${\mathcal P}(\overline{{\mathcal Z}^{*}})$ is {\em the coarsest equitable partition finer than ${\mathcal P}({\mathcal Z}^{*})$.} 

	In the subsection, we establish a one-to-one correspondence ${\mathcal P}$ between $M$-invariant dual subspaces and equitable partitions of a BN's STG. Moreover, we prove the isomorphism between two lattices: $({\mathcal X}^{**},\subseteq)$ and $(\Pi,\preceq)$. This allows us to study the inclusion relation between dual subspaces from a partition perspective. Theorem \ref{t3.3} gives a graphical representation of $M$-invariant dual subspaces. Additionally, we examine the smallest $M$-invariant dual subspace generated by a given dual subspace from a graphical standpoint. Based on these results, we obtain a complete structural characterization of $M$-invariant dual subspaces, which will be presented in the subsequent subsection. 

\subsection{Structures of \texorpdfstring{$M$}--invariant dual subspaces}\label{III-B}

Given a dual subspace ${\mathcal Z}^{*}$, let $\overline{{\mathcal Z}^{*}}$ be the smallest $M$-invariant dual subspace containing ${\mathcal Z}^{*}$. As concluded in Subsection \ref{III-AA}, the partition ${\mathcal P}(\overline{{\mathcal Z}^{*}})$ of the STG is the {\em{coarsest}} equitable partition finer than ${\mathcal P}({\mathcal Z}^{*})$. We define ${\mathcal P}(\overline{{\mathcal Z}^{*}})$ as the equitable partition generated by ${\mathcal P}({\mathcal Z}^{*})$.

For a general digraph ${\mathcal G}$ and a partition $\pi$ of $\mathcal G$, we define $E\pi$ as the coarsest equitable partition finer than $\pi$. We call $E\pi$ the equitable partition generated by $\pi$. In the specific case where ${\mathcal G}$ is the STG of a BN, if $\pi={\mathcal P}({\mathcal Z}^{*})$ then $E\pi={\mathcal P}(\overline{{\mathcal Z}^{*}})$. 

\begin{lemm}
	\label{new-l3.8}
	Given a digraph ${\mathcal G}$, let $\pi_1$ and $\pi_2$ be two partitions of ${\mathcal G}$. If $\pi_1 \preceq E\pi_2$, then $E\pi_1 \preceq E\pi_2$. Moreover, if $\pi_1 \preceq \pi_2$, then $E\pi_1 \preceq E\pi_2$.
\end{lemm}	

\begin{proof}
	As illustrated above, $E\pi_1$ is the coarsest equitable partition finer than $\pi_1$. It follows that 
		$E\pi_1 \preceq E\pi_2$ from $\pi_1 \preceq E\pi_2$.
		Moreover, if $\pi_1 \preceq \pi_2$, then $\pi_1 \preceq E\pi_2$. We get $E\pi_1 \preceq E\pi_2$.
\end{proof}

Before proceeding further, let us introduce an operation of shrinking. Consider a digraph ${\mathcal G}$ and
$C\subseteq V({\mathcal G})$. To shrink $C$ means to merge all vertices of $C$ into a single vertex and then add a self-loop to the new vertex if an edge exists between these vertices. We denote the resulting digraph as ${\mathcal G}/C$ and the new vertex as $c$. In ${\mathcal G}/C$, the edges between the new vertex $c$ and the vertices in $V({\mathcal G})\setminus C$ are inherited from the edges of ${\mathcal G}$. Note that ${\mathcal G}/C$ may generally have multiple edges between some pair of vertices. We replace multiple edges with a single edge.

Let $\pi:=\{C_1,\ldots,C_k\}$ be a partition of $V({\mathcal G})$. Consider $C$ which is a subset of $C_1$. Define a quotient partition of $\pi$ induced by $C$ as $\pi/C=\{C_1/C, C_2, \ldots, C_k\}$, where $C_1/C=\{c\}\cup (C_1\backslash C)$.

If $\pi$ is equitable, we can simultaneously shrink each cell to a new vertex. The resulting digraph is denoted by ${\mathcal G}/\pi$. In fact, ${\mathcal G}/\pi$ is the quotient digraph of ${\mathcal G}$ over $\pi$ as defined in Subsection \ref{section2.1.2}. 
Furthermore, for any partition $\pi_1$ such that $\pi_1\preceq \pi$, there exists a quotient partition $\pi_1/\pi$ of $\pi_1$ obtained by shrinking certain subsets of the cells of $\pi_1$, where these subsets correspond to all the cells of $\pi$. We call $\pi_1/\pi$ a quotient partition of $\pi_1$ induced by $\pi$. Hence, $\pi_1/\pi$ can be regarded as a partition of ${\mathcal G}/\pi$. Examples of the shrinking operation are illustrated in Fig. \ref{fig:example}. Fig. \ref{fig:example-a} shows a given STG, while Fig. \ref{fig:example-b} depicts its quotient digraph for the equitable partition $\pi:=\{\{v_1,v_4\},\{v_2,v_3\},\{v_5\},\{v_6,v_7\},\{v_8\}\}$. For Fig. \ref{fig:example-a}, there exists a partition $\pi_1\!\!:=\!\!\{\!\{v_1,v_4,v_5\},\!\{v_2,v_3\},\!\{v_6,v_7\},$ $\{v_8\}\}$ satisfying $\pi_1\preceq\pi$. In Fig. \ref{fig:example-b}, the quotient partition of $\pi_1$ is $\pi_1/\pi:=\{\{v_{1,4},v_5\},\{v_{2,3}\},\{v_{6,7}\},\{v_8\}\}$. The vertices in Fig. \ref{fig:example-a} and Fig. \ref{fig:example-b} are color-coded to represent the distinct cells of $\pi_1$ and $\pi_1/\pi$, respectively. 


We first recall the Algorithm 3.11 from \citep{che21}, which determines the smallest $M$-invariant dual subspace $\overline{{\mathcal Z}^{*}}$ containing the given dual subspace ${\mathcal Z}^{*}$. Suppose ${\mathcal Z}^{*}={\mathcal F}_{\ell}\{z_0:\sim G_0{\bar x}\}$. Algorithm 3.11 iteratively extends the subspace as follows: ${\mathcal F}_{\ell}\{z^{i+1}\}={\mathcal F}_{\ell}\{z^{i}\cup\{z_{i+1}\}\}$, $i=0,1,\ldots$, where $z^0=\{z_0\}$ and $z_i(x):\sim G_i{\bar x}=G_0M^i{\bar x}$, $i=1,\ldots$. The process terminates at step $k$ if 
\begin{equation}
	\label{space-eq}
	{\mathcal F}_{\ell}\{z^k\}={\mathcal F}_{\ell}\{z^{k+1}\}.
\end{equation}
We get ${\mathcal P}({\mathcal F}_{\ell}\{z^i\})\preceq{\mathcal P}({\mathcal F}_{\ell}\{z^{i+1}\})$, $i=0,\ldots,k-1$. And the equality holds iff $i\geq k$.

From the above algorithm, ${\mathcal P}(\overline{{\mathcal Z}^{*}})$ has the following properties.
\begin{cor}
	\label{cor-3.9}
	Consider a STG ${\mathcal G}$ and a dual subspace ${\mathcal Z}^{*}$.
	If two states share the same out-neighbor and belong to the same cell of ${\mathcal P}({\mathcal Z}^{*})$, then they are also in the same cell of ${\mathcal P}(\overline{{\mathcal Z}^{*}})$.
\end{cor}	
\begin{proof}
Let the vertex set of the STG be $\{v_1,v_2,\ldots,v_{2^n}\}$.
If $v_1$ and $v_2$ share the same out-neighbor, then the partition $E\pi:=\{\{v_1, v_2\},\{v_3\},\ldots,\{v_{2^n}\}\}$ is equitable. Moreover, if $v_1$ and $v_2$ are also in the same cell ${\mathcal P}({\mathcal Z}^{*})$, then ${\mathcal P}({\mathcal Z}^{*})\preceq E\pi$. Consequently, by Lemma \ref{new-l3.8}, we have ${\mathcal P}(\overline{{\mathcal Z}^{*}})\preceq E\pi$. We can therefore conclude that $v_1$ and $v_2$ are in the same cell of ${\mathcal P}(\overline{{\mathcal Z}^{*}})$.
\end{proof}

 Suppose $v_1,v_2 \in \Delta_{2^n}$ are two vertices in the STG that share the same out-neighbor and are in the same cell of ${\mathcal P}({\mathcal Z}^{*})$. From the proof of Corollary \ref{cor-3.9}, we get that ${\mathcal P}({\mathcal Z}^{*})\preceq{\mathcal P}(\overline{{\mathcal Z}^{*}})\preceq E\pi$, where $E\pi=\{\{v_1, v_2\},\{v_3\},\ldots,\{v_{2^n}\}\}$.

Let us define ${\mathcal G}_0^1:=STG/E\pi$ and $\pi_0^1:={\mathcal P}({\mathcal Z}^{*})/E\pi$. Then, $E\pi_0^1:={\mathcal P}(\overline{{\mathcal Z}^{*}})/\{v_1,v_2\}$.
Analogous to the proof of Corollary \ref{cor-3.9}, if two vertices in ${\mathcal G}_0^1$ share the same out-neighbor and belong to the same cell of $\pi^1_{0}$, they necessarily belong to the same cell of $E\pi^1_{0}$. Thus, we can recursively apply the shrinking process to the in-neighbors of vertices in ${\mathcal G}_0^1$. In the new digraph resulting from shrinking vertices of ${\mathcal G}_0^1$, $\pi^1_{0}$ will induce a new partition. We can similarly perform the above shrinking operation on this new digraph and partition. The iterative process of simplifying the graph structure and its corresponding partitions discussed above is systematically presented and summarized in Algorithm \ref{alg-1}.

\begin{algorithm}
	\caption{Simplifying the STG of a BN~\eqref{2.2.3} and a partition ${\mathcal P}({\mathcal Z}^{*})$, where ${\mathcal Z}^*$ is a dual subspace.}
	\label{alg-1}
	\begin{algorithmic}[1]
		\Function {Shrinking}{STG, ${\mathcal P}({\mathcal Z}^{*})$}
        \State {${\mathcal{G}}\gets$ the STG}
        \State {$\pi_{0}\gets{\mathcal P}({\mathcal Z}^{*})$}
		\While  {there are vertices $v,u\in V({\mathcal{G}})$ sharing the same out-neighbor belonging to the same cell of $\pi_{0}$}
        \State {${\mathcal{G}}\gets{\mathcal{G}}/\{u,v\}$}
        \State {$\pi_{0}\gets\pi_{0}/\{u,v\}$}
        \EndWhile
        \State\Return {(${\mathcal{G}}$, $\pi_{0}$)}
		\EndFunction
	\end{algorithmic}
\end{algorithm}

In the resultant digraph ${\mathcal{G}}$ of Algorithm \ref{alg-1}, the in-degree of each vertex is bounded above by the cardinality of ${\mathcal P}({\mathcal Z}^{*})$. For illustration, Fig. \ref{fig:example-b} is derived from Fig. \ref{fig:example-a} through the application of this SHRINKING operation.

\begin{lemm}
\label{g is quotient}
  Suppose (${\mathcal{G}}$, $\pi_{0}$)=SHRINKING(STG, ${\mathcal P}({\mathcal Z}^{*})$). ${\mathcal{G}}$ is the quotient digraph of the STG corresponding to an equitable partition. 
\end{lemm}

\begin{proof}
    Let $G^0, G^1, \ldots, G^m$ denote the sequence of digraphs generated during the execution of the SHRINKING operation, where $G^0$ is the initial STG; for each $i\in [0;m-1]$, $G^{i+1} = G^i / \{u_i, v_i\}$, where $u_i, v_i\in V(G^i)$ represents the pair of vertices merged at each iteration; $G^m = G$ is the final output digraph. We show, on induction on $i$, that each $G^i$ is the quotient digraph of the STG corresponding to an equitable partition.

    We know that each vertex of $G^i$ corresponds to a vertex subset of the STG and $G^i$ corresponds to a partition $\pi^i$ of the vertex set of the STG. 
    Clearly, $G^0$ corresponds to the trivial equitable partition $E\pi(STG)$. 
    Suppose $\pi^i:=\{C_1,C_2,\ldots,C_m\}$ is equitable. Consider $G^{i+1} = G^i / {u_i, v_i}$. Let $u_i$, $v_i$ and their same out-neighbor correspond to cells $C_1$, $C_2$ and $C_3$, respectively. Then, the out-neighbors of vertices in $C_1$ and $C_2$ are all in $C_3$. 
    Thus, $\pi^{i+1}:=\{C_1\cup C_2,C_3,\ldots,C_m\}$ is equitable. 

    Therefore, ${\mathcal{G}}$ is the quotient digraph of the STG corresponding to an equitable partition. 
    \end{proof}

Let ${\mathcal{G}}$ correspond to the equitable partition $\pi$, That is, ${\mathcal G}=STG/\pi$. According to the \textbf{while}-condition in Algorithm \ref{alg-1} and Corollary \ref{cor-3.9}, we know that ${\mathcal P}({\mathcal Z}^{*})\preceq{\mathcal P}(\overline{{\mathcal Z}^{*}}) \preceq\pi$, $\pi_0={\mathcal P}({\mathcal Z}^{*})/\pi$ and $E\pi_0={\mathcal P}(\overline{{\mathcal Z}^{*}})/\pi$. 
 In a special case where $E\pi_0$ is trivial, it follows that ${\mathcal P}(\overline{{\mathcal Z}^{*}})=\pi$ and ${\mathcal G}$ is the quotient digraph of the STG corresponding to ${\mathcal P}(\overline{{\mathcal Z}^{*}})$. The subsequent discussion about the structural characteristics of $E\pi_{0}$ is based on this observation.



Before our further analysis, we first give the following properties of ${\mathcal G}$.
\begin{lemm}\label{l3.6}
For any given BN, every component of its STG contains a unique directed cycle or loop.
\end{lemm}
\begin{proof}
	Without loss of generality, we assume the STG is connected. Given that each vertex has an out-degree of 1, the underlying graph of the STG contains exactly one undirected cycle, where a loop is considered as a cycle of length 1. Since a trajectory in a BN eventually converges to a directed cycle \citep{che12}, the directed cycle is unique in STG.	
 \end{proof}

As mentioned before, for any equitable partition of the STG, its corresponding quotient digraph represents the STG of the associated dual dynamics. Thus, by virtue of Lemma \ref{l3.6}, we can assert that each component of the quotient digraphs of the STG contains precisely one directed cycle or loop.

By Lemma \ref{g is quotient}, the obtained digraph ${\mathcal G}$ in Algorithm \ref{alg-1} is a quotient digraph of the STG. Thus, every component of ${\mathcal G}$ contains a unique directed cycle or loop.
 Without loss of generality, we focus our analysis on connected ${\mathcal G}$. We further investigate the structural characteristic of ${\mathcal P}(\overline{{\mathcal Z}^{*}})$ through the analysis of $E\pi_{0}$ by considering two distinct cases: when it contains a loop and when ${\mathcal G}$ contains a cycle. 

\subsubsection{The case that the digraph \texorpdfstring{${\mathcal G}$}- contains a loop}
\label{subsectionB-1}


\begin{lemm}
\label{path}
Let $\widetilde{\mathcal{G}}$ be a digraph where each vertex with in-degree of 1. Suppose $\widetilde{\mathcal{G}}$ contains a loop and $v_1$ is the root incident with the loop. For the partition $\pi_{\{v_1\}}:=\{\{v_1\}, V(\widetilde{\mathcal{G}})\setminus \{v_1\}\}$, we have $E\pi_{\{v_1\}}=\{\{v_1\}, N_{in}(v_1,1), \ldots, N_{in}(v_1, \operatorname{dist}_{in}(v_{1}))\}$.

\end{lemm}
\begin{proof}
Let (${\mathcal{G}}$, ${\pi_{0}}$)=SHRINKING($\widetilde{\mathcal{G}}$, $\pi_{\{v_1\}}$). By Lemma \ref{g is quotient}, ${\mathcal{G}}$ corresponds to an equitable partition $\pi:=\{C_1,C_2,\ldots,C_m\}$ of $\widetilde{\mathcal{G}}$.
From Algorithm \ref{alg-1}, we know that $v_1$ is still the root of ${\mathcal{G}}$ and ${\pi_{0}}=\{\{v_1\},V({\mathcal{G}})\setminus\{v_1\}\}$. Since ${\pi_{0}}\preceq \pi$, $v_1$ forms a single cell of $\pi$, denoted as $C_1$. Denote $C_2$ as the cell containing a vertex adjacent to $v_1$. Since $\pi$ is equitable, all vertices in $C_2$ are adjacent to $v_1$. Thus, $C_2\subset N_{in}(v_1,1)$. On the other hand, $N_{in}(v_1,1)$ is contained in one cell of $\pi$ from Algorithm \ref{alg-1}. We now get that $N_{in}(v_1,1)=C_2$. 

By induction on the distance $i$ to $v_1$, it can be readily verified that $N_{in}(v_{1}, i):=C_{i+1}\}$. That is, $\pi:=\{\{v_1\}, N_{in}(v_1,1), \ldots, N_{in}(v_1, \operatorname{dist}_{in}(v_{1}))\}$. Then, ${\mathcal{G}}$, the quotient digraph corresponds to $\pi$, is a path. We conclude that the trivial partition of ${\mathcal{G}}$ is $E{\pi_{0}}$, the coarsest equitable partition finer than ${\pi_{0}}$. Thus, $E\pi_{\{v_1\}}=\pi$.


\end{proof}

\begin{them}
\label{thm-loop}
Suppose (${\mathcal{G}}$, $\pi_{0}$)=SHRINKING(STG, ${\mathcal P}({\mathcal Z}^{*})$). If ${\mathcal G}$ contains a loop, then ${\mathcal G}$ is precisely the quotient digraph of the STG corresponding to ${\mathcal P}(\overline{{\mathcal Z}^{*}})$.
\end{them}

\begin{proof}
Denote the loop contained in ${\mathcal G}$ by $e$. Since $e$ is the unique loop (Lemma \ref{l3.6}) and every state in ${\mathcal G}$ can reach the end $v_1$ of $e$ \citep{che12}, we call $v_1$ the root of ${\mathcal G}$. 

If $|V ({\mathcal G})|>1$, then the root $v_1$ has in-neighbors. Since vertices with the same out-neighbor in ${\mathcal G}$ belong to different cells of $\pi_{0}$, this property is preserved in $E\pi_{0}$ according to $\pi_{0}\preceq E\pi_{0}$. Thus, $v_1$ and its in-neighbors are in different cells of $E\pi_{0}$. 

We first prove that $v_1$ forms a singleton cell of $E\pi_{0}$.

Let $C_1\in E\pi_{0}$ be the cell containing $v_1$. Since $E\pi_{0}$ is equitable, each vertex in the cell $C_1$ must have an out-neighbor in $C_1$, as $v_1$ does. Suppose, for contradiction, that $|C_1|>1$. Let $v_i$ be the vertex closest to $v_1$ in $C_1$. Since $v_1$ and its in-neighbors are in different cells of $\pi_{0}$, it follows that $v_i$ is not an in-neighbor of $v_1$. Thus, the out-neighbor of $v_i$, which is also in $C_1$, is closer than $v_i$. This contradicts the ``closest" property of $v_i$. Therefore, $v_1$ indeed forms a singleton cell of $E\pi_{0}$. It follows that $E\pi_{\{v_1\}}\preceq E\pi_{0}$, where $\pi_{\{v_1\}}:=\{\{v_1\}, V(\widetilde{\mathcal{G}})\setminus \{v_1\}\}$. From Lemma \ref{path}, we have $E\pi_{\{v_1\}}=\{\{v_1\}, N_{in}(v_1,1), \ldots, $ $N_{in}(v_1, \operatorname{dist}_{in}(v_{1}))\}$. We can conclude that the vertices in the same cell of $E\pi_{0}$ are with the same distance to $v_1$ in ${\mathcal G}$. 

We now prove that $E\pi_{0}$ is trivial for ${\mathcal G}$. By contradiction, we suppose $E\pi_{0}$ is non-trivial. Let $C_2$ be the non-singleton cell closest to $v_1$. Consider vertices $v_i$ and $v_j$ in $C_2$. Their out-neighbors must be distinct. By the definition of the equitable partition, these out-neighbors must be in a non-singleton cell, which is closer than $C_2$. However, this contradicts the ``closest" property of $C_2$. Therefore, we conclude that $E\pi_{0}$ is trivial. Consequently, ${\mathcal G}$ is exactly the quotient digraph of the original STG corresponding to the equitable partition ${\mathcal P}(\overline{{\mathcal Z}^{*}})$.
\end{proof}

\begin{exa} \label{e-3}
Consider the following BN:
\begin{equation}
\left\{
\begin{aligned}
x_1(t+1)=&[x_1(t)\wedge x_2(t)]\vee\\
&[\neg x_1(t)\wedge [\neg x_2(t)\vee x_3(t)]],\\
x_2(t+1)=&x_1(t)\vee x_2(t)\vee x_3(t),\\
x_3(t+1)=& x_2(t)\vee[\neg x_1(t) \wedge x_3(t)].
\end{aligned}
\right.
\end{equation}
Its ASSR is calculated as 
\begin{equation}
x(t+1)=Mx(t),
\end{equation}
where
\begin{equation*}
\begin{array}{l l llllll}
M=\delta_{8}[1 &1 &6 &6 &1 &5 &1 &4].
\end{array}    
\end{equation*}
Its STG is depicted in Fig. \ref{fig:example-a}, where $v_1$ is the root. 

Consider the dual subspace ${\mathcal Z}^{*}$ with the structure matrix $\begin{array}{l l llllll}
G=\delta_{4}[1 &2 &3 &3 &2 &1 &1 &4]
\end{array}$. Then, ${\mathcal P}({\mathcal Z}^{*})=\{\{v_1,v_4,v_5\},\{v_2,v_3\},\{v_6,v_7\},\{v_8\}\}$. In Fig. \ref{fig:example-a}, vertices are color-coded to represent the distinct cells of ${\mathcal P}({\mathcal Z}^{*})$.
Upon applying Algorithm \ref{alg-1} to the STG and ${\mathcal P}({\mathcal Z}^{*})$, we shrink $\{v_1,v_4\}$, $\{v_2,v_3\}$ and $\{v_6,v_7\}$ in the STG to new vertices $v_{1,4}$, $v_{2,3}$ and $v_{6,7}$, respectively. Fig. \ref{fig:example-b} shows the resultant digraph ${\mathcal G}$, while the partition $\pi_0=\{\{v_{1,4},v_5\},\{v_{2,3}\},\{v_{6,7}\},\{v_8\}\}$ is obtained from ${\mathcal P}({\mathcal Z}^{*})$. According to Theorem \ref{thm-loop}, ${\mathcal G}$ is exactly the quotient digraph of the original STG corresponding to the equitable partition ${\mathcal P}(\overline{{\mathcal Z}^{*}})$. We obtain
$$
    {\mathcal P}(\overline{{\mathcal Z}^{*}})=\{\{v_1,v_4\},\{v_2,v_3\},\{v_5\},\{v_6,v_7\},\{v_8\}\}.
$$
\end{exa}
\begin{figure}
\centering

\subcaptionbox{The original STG. Here, vertices are color-coded to represent the distinct cells of ${\mathcal P}({\mathcal Z}^{*})$.\label{fig:example-a}}
{%
		\begin{tikzpicture}[->,thick]
				\tikzstyle{every node}=[draw,circle,radius=0.2mm,inner sep=0pt,minimum size=1em];
				\def \n {8}
				\def \radius {0.7cm}
				\def \margin {25}
				\node (v1) [fill=c2]at (3.2, 1){$1$};
				\put(108.5,25){$e$}
				\node (v2)[fill=d4] at (2.0, 1.5){$2$};
				\node (v3)[fill=d4] at (2.0, 1){$3$};
				\node (v4)[fill=c2] at (2.0, 0.5){$4$};
				\node (v5)[fill=c2] at (0.8, 1){$5$};
				\node (v6)[fill=d2] at (-0.4,1.5){$6$};
				\node (v7)[fill=d2] at (-0.4,1){$7$};
				\node (v8)[fill=d1] at (-1.6,1){$8$};
				\path (v8) edge (v7);
				\path (v7) edge (v5);
				\path (v6) edge (v5);
				\path (v5) edge (v3);
				\path (v4) edge (v1);
				\path (v3) edge (v1);
				\path (v2) edge (v1);
				\draw (v1.-60) arc (210:180+300:2mm);
			\end{tikzpicture}
	}
\subcaptionbox{(${\mathcal G}, \pi_0$)=SHRINKING (STG, ${\mathcal P}({\mathcal Z}^{*})$), where  $v_{1,4}$, $v_{2,3}$ and $v_{6,7}$ are obtained by shrinking $\{v_1,v_4\}$, $\{v_2,v_3\}$ and $\{v_6,v_7\}$, respectively. Moreover, $\pi_0=\{\{v_{1,4},v_5\},\{v_{2,3}\},\{v_{6,7}\},\{v_8\}\}$.\label{fig:example-b}}
{%
		\begin{tikzpicture}[->,thick]
				\tikzstyle{every node}=[draw,circle,radius=0.2mm,inner sep=0pt,minimum size=1em];
				\def \n {8}
				\def \radius {0.7cm}
				\def \margin {25}
				\node (v1) [fill=c2,label={[shift={(0.0,-1.0)}]$v_{1,4}$}]at (3.2, 1){};
				\put(105,25){$e$}
				\node (v3)[fill=d4,label={[shift={(0.0,-0.0)}]$v_{2,3}$}] at (2, 1){};
				\node (v5)[fill=c2] at (0.8, 1){$5$};
				\node (v7)[fill=d2,label={[shift={(0.0,-1.0)}]$v_{6,7}$}] at (-0.4,1){};
				\node (v8)[fill=d1] at (-1.6,1){$8$};
				\path (v8) edge (v7);
				\path (v7) edge (v5);
				\path (v5) edge (v3);
				\path (v3) edge (v1);
				\draw (v1.-60) arc (210:180+300:2mm);
			\end{tikzpicture}
	}
\caption{Given a connected STG and a dual subspace ${\mathcal Z}^{*}$, the process of finding ${\mathcal P}(\overline{{\mathcal Z}^{*}})$, the coarsest equitable partition finer than ${\mathcal P}({\mathcal Z}^{*})$ is illustrated, where ${\mathcal P}({\mathcal Z}^{*})=\{\{v_{1},v_{4},v_5\},\{v_{2},v_{3}\},\{v_{6,7}\},\{v_8\}\}$. We finally get ${\mathcal P}(\overline{{\mathcal Z}^{*}})=\{\{v_1,v_7\},\{v_2,v_5\},\{v_6\},\{v_3,v_4\},\{v_8\}\}$, whose quotient digraph is illustrated in (b).}
\label{fig:example}
\end{figure}

\subsubsection{The case that \texorpdfstring{${\mathcal G}$}- contains a cycle}
\label{cycle}
Now, we consider the case where the simplified ${\mathcal G}$, obtained from applying Algorithm \ref{alg-1} to the STG and ${\mathcal P}({\mathcal Z}^{*})$, contains a cycle ${\mathcal C}$. 

We first focus on the equitable partitions of cycles.

\begin{them}{\label{l3.10}}
	Assume that ${\mathcal C}$ is a directed cycle of length $l$ and $v_1\in V({\mathcal C})$. A partition $\pi$ of ${\mathcal C}$ is equitable iff there exists a factor $q$ of $l$ such that $\pi=\{\{v_i | \operatorname{dist}(v_i\rightarrow v_1)\equiv j(\!\!\!\!\mod q)\}|j\in[0;q-1]\}$. Moreover, the quotient digraph ${\mathcal C}/\pi$ is a directed cycle of length $q$.
\end{them}
\begin{proof}
	It is easy to show that $\pi=\{\{v_i | \operatorname{dist}(v_i\rightarrow v_1)\equiv j(\!\!\!\!\mod q)\} | j\in[0;q-1]\}$ is equitable for any $q|l$. Its quotient digraph ${\mathcal C}/\pi$ is a directed $q$-cycle.
	
	Suppose $\pi=\{C_1,C_2,\ldots,C_q\}$ is a nontrivial equitable partition of ${\mathcal C}:v_1v_2\cdots v_l$. Assume $v_1\in C_1$ and its out-neighbor $v_2\in C_1$. By the definition of equitable partitions, their respective out-neighbors $v_2$ and $v_3$ must also be in the same cell. Thus, $v_3 \in C_1$. Proceeding inductively, we ultimately get $C_1=V({\mathcal C})$. 
	
	Assume $\pi$ is nontrivial and $V({\mathcal C})\notin\pi$. Then, by our previous conclusion, no two adjacent vertices can be in the same cell of $\pi$. As we concluded after Lemma \ref{l3.6}, the quotient digraph ${\mathcal C}/\pi$ is unicyclic, which is a cycle since $\pi$ contains multiple cells. Denoted ${\mathcal C}/\pi$ as $C_1C_2\cdots C_q$. In ${\mathcal C}/\pi$, the out-neighbors of vertices in $C_{i}$ are contained in $C_{i+1}$, $i\in [1;q-1]$ and the out-neighbors of vertices in $C_q$ are contained in $C_1$. We get that $\pi=\{\{v_i | \operatorname{dist}(v_i\rightarrow v_1)\equiv j(\!\!\!\!\mod q)\}| j\in[0;q-1]\}$ and $|C_1|=|C_2|=\cdots=|C_q|=l/q$, which follows $q|l$.
\end{proof}

Depending on whether ${\mathcal G}$ is a single cycle or not, we further discuss the structures of $E\pi_{0}$ in two cases.

\noindent{\bf Case 1.} The digraph ${\mathcal G}$ is a single cycle.

(1) If there exists a proper factor $q$ of $l$ such that $\pi_0$ is coarser than the equitable $q$-partition $E\pi:=\{\{v | \operatorname{dist}(v\rightarrow v_1)\equiv j(\mod q)\}|j\in [0;q-1]\}$, i.e., $\pi_0\preceq E\pi$, then $E\pi_0\preceq E\pi$ and $E\pi_0$ is non-trivial. In the quotient $q$-cycle ${\mathcal G}_1:={\mathcal G}/E\pi$, let $\pi_1:=\pi_0/E\pi$ be the quotient partition of $\pi_0$ induced by $E\pi$. We can determine the quotient cycle of ${\mathcal G}$ corresponding to $E\pi_0$ by determining the quotient cycle of ${\mathcal G}_1$ corresponding to $E\pi_{1}$. (2) Otherwise, $E\pi_0$ is trivial for ${\mathcal G}$. Thus, ${\mathcal G}$ is the quotient digraph of the STG corresponding to ${\mathcal P}(\overline{{\mathcal Z}^{*}})$. This iterative process is formalized in Algorithm \ref{alg-2}.

\begin{algorithm}
	\caption{The algorithm for determining the quotient digraph of the STG corresponding to ${\mathcal P}(\overline{{\mathcal Z}^{*}})$ when ${\mathcal{G}}$ is a cycle.}
	\label{alg-2}
	\begin{algorithmic}[1]
		\Function {Cycle}{${\mathcal{G}}$, $\pi_0$}
		\If {there exists a proper factor $q$ of $|{\mathcal{G}}|$ such that $\pi_0\preceq E\pi:=\{\{v | \operatorname{dist}(v\rightarrow v_1)\equiv j(\!\!\!\!\mod q)\}|j\in [0;q-1]\}$}
        \State {${\mathcal{G}}\gets {\mathcal{G}}/E\pi$} 
        \State {$\pi_0\gets\pi_0/E\pi$, which is the quotient partition of $\pi_0$ induced by $E\pi$}
        \State\Return {CYCLE(${\mathcal{G}}$, $\pi_0$)}
        \Else 
        \State\Return {(${\mathcal{G}}$, $\pi_0$)} \Comment{${\mathcal{G}}$ is the quotient digraph.}
        \EndIf
		\EndFunction
	\end{algorithmic}
\end{algorithm}

\noindent{\bf Case 2.} The digraph ${\mathcal G}$ is unicyclic, containing a unique cycle ${\mathcal C}$ of length $l$.

Case 2.1. All vertices in ${\mathcal C}$ are partitioned into one cell of $\pi_{0}$. 

 During the execution of Algorithm 3.11 \citep{che21}, all vertices in ${\mathcal C}$ always produce the same in $z_i(x)\sim G_i{\bar x}$ for all $i\in[0;k]$ since $V({\mathcal C})$ is a subset of a single cell in $\pi_{0}$. As a result, $V({\mathcal C})$ is contained within a single cell of $E\pi_{0}$. Given this observation, we shrink $V({\mathcal C})$ in ${\mathcal G}$ to a new vertex, denoted as $v$, which is incident with a loop. In ${\mathcal G}/V({\mathcal C})$, $\pi_{0}/V({\mathcal C})$ is the quotient partition induced from $\pi_{0}$.
 By Theorem \ref{thm-loop}, the quotient digraph of the original STG corresponding to ${\mathcal P}(\overline{{\mathcal Z}^{*}})$ is obtained from SHRINKING(${\mathcal G}/V({\mathcal C})$,$\pi_{0}/V({\mathcal C})$).

Case 2.2. Consider the case where $V({\mathcal C})$ is partitioned into distinct cells of $\pi_{0}$. The induced partition $\pi'_{0}$ on the vertices of ${\mathcal C}$ is defined as: $\pi'_{0} := \{C_i \cap V({\mathcal C}) \mid \forall C_i \in \pi_{0}, C_i \cap V({\mathcal C}) \neq \emptyset\}$.

(1) If $E\pi'_{0}$, the equitable partition of ${\mathcal C}$ generated by $\pi'_{0}$, is trivial, then ${\mathcal P}(\overline{{\mathcal Z}^{*}})$ exhibits specific characteristics as detailed in the following theorem.

\begin{them}
\label{the-5}
If $E\pi'_{0}$ is trivial, then $E\pi_{0}$ is trivial and ${\mathcal G}$ is the quotient digraph of the STG corresponding to ${\mathcal P}(\overline{{\mathcal Z}^{*}})$.
\end{them}
\begin{proof}
We first prove that the vertices in $V({\mathcal C})$ form singleton cells in $E\pi_{0}$. 

Given that the partition of a vertex is solely determined by the set of vertices within its reachability, the partition of ${\mathcal C}$ in $E\pi_{0}$ is consistent with the equitable partition $E\pi'_{0}$ of ${\mathcal C}$ generated by $\pi'_{0}$. Since $E\pi'_{0}$ is trivial, each vertex in ${\mathcal C}$ belongs to a distinct cell of $E\pi_{0}$. Let $C_1,C_2,\ldots,C_l$ denote the cells in $E\pi_{0}$ containing the vertices in $V({\mathcal C})$, respectively.  
In the quotient digraph ${\mathcal G}/E\pi_{0}$, the subgraph induced by the vertices corresponding to $C_1, C_2,\ldots,C_l$ is also an $l$-cycle.

By contradiction, we suppose that a vertex $v_i\in V({\mathcal C})$ shares a cell in $E\pi_{0}$ with a vertex $v_{j}\notin V({\mathcal C})$. Let $v'$ be the vertex in $V({\mathcal C})$ closest to $v_j$. On the path from $v_j$ to $V'$, let $v_{j+1}$ be the in-neighbor of $v'$ (refer to Figure \ref{fig:enter-label}). Thus, $v_{j+1}\notin V({\mathcal C})$. By the definition of ${\mathcal G}$ and $\pi_0$, the out-neighbors of $v_i$ and $v_j$ are distinct. We know that their out-neighbors belong to the same cell in $E\pi_{0}$ by the definition of equitable partition. Furthermore, there exists $v_{i+1}\in V({\mathcal C})$ such that each vertex in the path from $v_i$ to $v_{i+1}$ shares a cell in $E\pi_{0}$ with the a vertex in the path from $v_j$ to $v_{j+1}$. As a result, $v_{i+1}$ and $v_{j+1}$ belong to the same cell in $E\pi_{0}$. 

Given that the in-neighbors of $v'$ are in different cells of $\pi_{0}$ and $E\pi_{0}$, the out-neighbor of $v_{i+1}$ cannot be $v'$. The distinct out-neighbors of $v_{i+1}$ and $v_{j+1}$, both in $V({\mathcal C})$, must belong to the same cell of $E\pi_{0}$. This contradicts our premise that the vertices in ${\mathcal C}$ are in distinct cells of $E\pi_{0}$. 
Therefore, we conclude that each vertex in $V({\mathcal C})$ must form a singleton cell in $E\pi_{0}$.

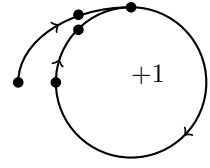
\begin{figure}
    \centering
   \begin{tikzpicture}[style=thick]
\coordinate (v) at (0,1);
\coordinate (vj) at (-1.5,0);
\coordinate (vi) at (-1,0);
\coordinate (vj+1) at (-0.7,0.9);
\coordinate (vi+1) at (-0.7,0.7);
\fill (v) circle (2pt);
\fill (vj) circle (2pt);
\fill (vi) circle (2pt);
\fill (vj+1) circle (2pt);
\fill (vi+1) circle (2pt);
\draw[thick, postaction={decorate, decoration={
    markings,
    mark=at position 0.5 with {\arrow{>}}
}}] (vj) arc[start angle=180, end angle=90, x radius=1.5cm, y radius=1cm];
\draw[thick, postaction={decorate, decoration={
    markings,
    mark=at position 0.2 with {\arrow{>}}
}}] (vi) arc[start angle=-180, end angle=-270, radius=1cm];
\draw[thick, postaction={decorate, decoration={
    markings,
    mark=at position 0.5 with {\arrow{>}}
}}] (v) arc[start angle=90, end angle=-180, radius=1cm];

\put(0,32){$v'$}
\put(-55,0){$v_j$}
\put(-26,0){$v_i$}
\put(-25,32){$v_{j+1}$}
\put(-15,18){$v_{i+1}$}
\end{tikzpicture}
    \caption{Proof of Theorem \ref{the-5}}
    \label{fig:enter-label}
\end{figure}

We now prove that each vertex not in $V({\mathcal C})$ also forms a singleton cell in $E\pi_{0}$.

Suppose, for contradiction, that two vertices $v_i, v_j \notin V({\mathcal C})$ belong to the same cell in $E\pi_{0}$. Then, for all $k\in[1;\operatorname{dist}^*_{out}(v_{i})]$, $N_{out}(v_{i}, k)$ and $N_{out}(v_{j}, k)$, which are distinct, must belong to the same cell in $E\pi_{0}$. Without loss of generality, there exists a $k$ such that at least one of $N_{out}(v_{i}, k)$ and $N_{out}(v_{j}, k)$ is in $V({\mathcal C})$. However, this contradicts our previous conclusion that each vertex in $V({\mathcal C})$ forms a singleton cell of $E\pi_{0}$.

\end{proof}

(2) If $E\pi'_{0}$ is non-trivial, we can first shrink all the cells of $E\pi'_{0}$ in the original digraph ${\mathcal G}$. In other words, we replace the original cycle ${\mathcal C}$ with its quotient cycle ${\mathcal C}/E\pi'_{0}$. Denote the resulting digraph by ${\mathcal G}_1$. In ${\mathcal G}_1$, let $\pi_1$ be the quotient partition of $\pi_0$ obtained by shrinking the vertex subset corresponding to all cells of $E\pi'_{0}$. After applying the SHRINKING operation to ${\mathcal G}_1$ and $ \pi_1$, we obtain ${\mathcal G}_2$ and $ \pi_2$.
Here, ${\mathcal G}_2$ and $ \pi_2$ satisfy the condition in Theorem \ref{the-5}. It follows that $E\pi_2$ is trivial and ${\mathcal G}_2$ is the quotient digraph of ${\mathcal P}(\overline{{\mathcal Z}^{*}})$.

In conclusion, Algorithm \ref{a3.2.1} is summarized for determining the quotient digraph of the STG corresponding to ${\mathcal P}(\overline{{\mathcal Z}^{*}})$.
\begin{algorithm}
	\caption{An algorithm for computing the smallest $M$-invariant dual subspace containing a given dual subspace when the STG is connected.}
	\label{a3.2.1}
	\noindent{\bf Input:}
	STG ${\mathcal{G}}$, dual subspace ${\mathcal Z}^{*}$\\ 
	\noindent{\bf Output:} ${\mathcal P}(\overline{{\mathcal Z}^{*}})$ and its quotient digraph 
	\begin{algorithmic}[1]
		\Ensure The STG is connected
        \Function {Partition}{${\mathcal{G}}$, ${\mathcal P}({\mathcal Z}^{*})$}
        \State {$({\mathcal{G}},\!\pi_0)\!\gets\!$ SHRINKING$({\mathcal{G}}$,${\mathcal P}({\mathcal Z}^{*}))$ \!}
        \Comment{Algorithm \ref{alg-1}.}
        \If {${\mathcal{G}}$ contains a loop}
        \State\Return {$({\mathcal{G}}, \pi_0)$}
        \Else \Comment{${\mathcal{G}}$ contains a cycle}
        \State{${\mathcal{C}}\gets$ the unique cycle in ${\mathcal{G}}$}
        \If {${\mathcal{G}}={\mathcal{C}}$} \Comment{Case 1.}
        \State\Return {CYCLE(${\mathcal{G}}$, $\pi_0$)} \Comment{Algorithm \ref{alg-2}.}
        \Else \Comment{Case 2.}
        \If {$V({\mathcal{C}})$ is a subset of one cell of $\pi_0$}\Comment{Case 2.1.}
        \State\Return {SHRINKING$({\mathcal{G}}/V({\mathcal{C}}),\pi_0/V({\mathcal{C}}))$}
        \Else \Comment{Case 2.2.}
        \State {$\pi'_{0}\gets$ the induced partition on $V({\mathcal C})$ from $\pi_{0}$} 
        \State {$({\mathcal{C}_1}, E\pi'_{0})\gets$ CYCLE$({\mathcal{C}}, \pi'_{0})$}
        
        \If {${\mathcal{C}_1}={\mathcal{C}}$} \Comment{Case 2.2 (1). }
        \State\Return $({\mathcal{G}}, \pi_0)$
        \Else \Comment{Case 2.2 (2). }
        \State {${\mathcal{G}}_1\gets$ replacing ${\mathcal{C}}$ in ${\mathcal{G}}$ by ${\mathcal{C}_1}$}
        \State {$\pi_1\gets$ the quotient partition of $\pi_0$ obtained by shrinking the vertex subset corresponding to all the cells of $E\pi'_0$}
        \State\Return SHRINKING$({\mathcal{G}}_1,\pi_1)$
        \EndIf
        \EndIf
        \EndIf
        \EndIf
        \EndFunction
	\end{algorithmic}
\end{algorithm}

In the following, we revisit the procedure in Algorithm \ref{a3.2.1} for the case where ${\mathcal G}$ contains a cycle through a concrete example.
\begin{exa}
\label{e-4}
Consider the following BN:
\begin{equation}
\left\{
\begin{aligned}
x_1(t+1)=&x_1(t)\vee[x_2(t)\wedge x_3(t)] ,\\
x_2(t+1)=&[x_1(t)\wedge[[x_2(t)\wedge[x_3(t)\vee x_4(t)]]\\
&\vee\neg[x_2(t)\vee[\neg x_3(t)\vee x_4(t)]]\vee\\
&[\neg x_2(t)]]]\vee\\
&[\neg x_1(t)\wedge[[\neg[x_2(t)\wedge x_3(t)]\\
&\wedge[x_2(t)\vee[x_3(t)\wedge x_4(t)]]]]],\\
x_3(t+1)=&[x_1(t)\wedge
[
[x_2(t)\wedge [
[x_3(t)\wedge x_4(t)]\vee\\
&\neg[x_3(t)\vee x_4(t)]
]]
\vee\\
&
[\neg x_2(t)\wedge
[x_3(t)\vee x_4(t)]
]
]
]
\vee
\\
&[\neg x_1(t)\wedge
[
[x_2(t)\wedge[x_3(t)\vee x_4(t)]
]\\
&\vee[\neg x_2(t)\wedge[\neg
[x_3(t)\wedge x_4(t)]\\
&\wedge[x_3(t)\vee x_4(t)]
]
]
]
],\\
x_4(t+1)=&[\neg x_1(t)\wedge x_2(t)\wedge x_3(t)]\vee \neg x_4(t).\\
\end{aligned}
\right.
\end{equation}
Its ASSR is 
$$x(t+1)=Mx(t),$$
where
\begin{equation*}
\begin{array}{ccccccccc}
M=\delta_{16}[&2 &3 &4 &5 &6 &1 &6 &7\\
&5 &5 &10 &11 &12 &13 &14 &15].
\end{array}    
\end{equation*}
Its STG is shown in Fig. \ref{1-a}. Note that the dotted edge between $v_{16}$ and $v_{12}$ means the path from $v_{16}$ to $ v_{12}$. 

Consider dual subspace ${\mathcal Z}^{*}$ with the structure matrix $
\begin{array}{cccccccccccccccccc}
G=\delta_{4}[&1 &2 &3 &1 &2 &3 &2 &1
&3 &3 &4 &4 &4 &4 &4 &4].
\end{array}
$
It follows that ${\mathcal P}({\mathcal Z}^{*})=\{\{v_1,v_4,v_8\},\{v_2,v_5,v_7\},\{v_3,v_6,v_9,v_{10}\},$ $\{v_{11},\ldots,v_{16}\}\}$. In Fig. \ref{1-a}, vertices are color-coded to represent the distinct cells of ${\mathcal P}({\mathcal Z}^{*})$.

As outlined in Algorithm \ref{a3.2.1}, we initially apply the SHRINKING operation to the original STG and ${\mathcal P}({\mathcal Z}^{*})$. The resulting ${\mathcal G}$ and $\pi_0$ are illustrated in Fig. \ref{1-b}. 
In ${\mathcal G}$, vertices $v_{4,8}$, $v_{5,7}$ and $v_{9,10}$ are obtained by shrinking $\{v_4,v_8\}$, $\{v_5,v_7\}$ and $\{v_9,v_{10}\}$, respectively. Moreover, $\pi_0=\{\{v_1,v_{4,8}\},\{v_2,v_{5,7}\},\{v_3,v_6,v_{9,10}\},\{v_{11},\ldots,v_{16}\}\}$. Here, ${\mathcal G}$ contains a 6-cycle ${\mathcal C}$.

Analogous to Case 2.2, vertices in ${\mathcal C}$ are partitioned into distinct cells in $\pi_{0}$. The partition induced by $V({\mathcal C})$ from $\pi_{0}$ is
$\pi'_0=\{\{v_1,v_{4,8}\},\{v_2,v_{5,6}\},\{v_3,v_6\}\}$. Based on Algorithm \ref{alg-2}, the quotient digraph ${\mathcal C}/E\pi'_0$ is a 3-cycle with vertices $v_{1,4,8}$, $v_{2,5,7}$ and $v_{3,6}$. By replacing the ${\mathcal C}$ with ${\mathcal C}/E\pi'_0$, we obtain ${\mathcal G}_1$ and $\pi_{1}$, as illustrated in Fig \ref{1-d}. 

Given that (${\mathcal G}_1, \pi_{1}$)= SHRINKING(${\mathcal G}_1, \pi_{1}$), we conclude that ${\mathcal G}_1$ is the quotient digraph of the equitable partition ${\mathcal P}(\overline{{\mathcal Z}^{*}})$.
Thus,
$$\begin{aligned}
    {\mathcal P}(\overline{{\mathcal Z}^{*}})=&\{\{v_{1},v_{4},v_{8}\},\{v_3,v_6\},\{v_{2},v_{5},v_{7}\},\{v_9,v_{10}\},\\
&\{v_{11}\},\ldots,\{v_{16}\}\}.
    \end{aligned}$$
\end{exa}	
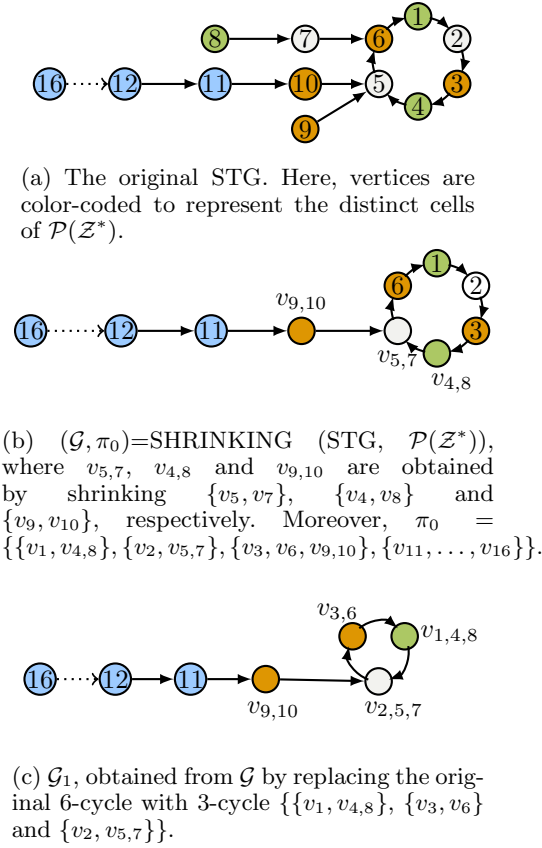
\begin{figure}
	\centering
	\setlength{\unitlength}{2mm}
	\subcaptionbox{The original STG. Here, vertices are color-coded to represent the distinct cells of ${\mathcal P}({\mathcal Z}^{*})$.\label{1-a}}
	{
		\begin{tikzpicture}[->,thick]
			\tikzstyle{every node}=[draw,circle,radius=0.2mm,inner sep=0pt,minimum size=1em];
			\def \n {12}
			\def \radius {0.6cm}
			\def \margin {45} 
			\node (v2)[fill=d4] at ({360/6 * (1 - 1)+30}:\radius) {$2$};
			\node (v1)[fill=d2] at ({360/6 * (2 - 1)+30}:\radius) {$1$};
			\node (v6)[fill=d1] at ({360/6 * (3 - 1)+30}:\radius) {$6$};
			\node (v5)[fill=d4] at ({360/6 * (4 - 1)+30}:\radius) {$5$};
			\node (v4)[fill=d2] at ({360/6 * (5 - 1)+30}:\radius) {$4$};
			\node (v3)[fill=d1] at ({360/6 * (6 - 1)+30}:\radius) {$3$};
			\node (v7)[fill=d4] at (-1.5,0.3) {$7$};
			\node (v8)[fill=d2] at (-2.7, 0.3) {$8$};
			\node (v9)[fill=d1] at (-1.5,-0.9) {$9$};
			\node (v10)[fill=d1] at (-1.5,-0.3) {$10$};
			\node (v11)[fill=c2] at (-2.7, -0.3) {$11$};
			\node (v12)[fill=c2] at (-3.9,-0.3) {$12$};
			\node (v16)[fill=c2] at (-4.9,-0.3) {$16$};
			
			\draw[->, >=latex] ({360/6 * (1 - 1)-30+\margin}:\radius)
			arc ({360/6 * (1 - 1)-30+\margin}:{360/6 * (1)-30-\margin}:\radius);
			\draw[->, >=latex] ({360/6 * (2- 1)-30+\margin}:\radius)
			arc ({360/6 * (2 - 1)-30+\margin}:{360/6 * (2)-30-\margin}:\radius);
			\draw[->, >=latex] ({360/6 * (3 - 1)-30+\margin}:\radius)
			arc ({360/6 * (3 - 1)-30+\margin}:{360/6 * (3)-30-\margin}:\radius);
			\draw[->, >=latex] ({360/6 * (4 - 1)-30+\margin}:\radius)
			arc ({360/6 * (4 - 1)-30+\margin}:{360/6 * (4)-30-\margin}:\radius);
			\draw[->, >=latex] ({360/6 * (5 - 1)-30+\margin}:\radius)
			arc ({360/6 * (5 - 1)-30+\margin}:{360/6 * (5)-30-\margin}:\radius);
			\draw[->, >=latex] ({360/6 * (6 - 1)-30+\margin}:\radius)
			arc ({360/6 * (6 - 1)-30+\margin}:{360/6 * (6)-30-\margin}:\radius);
			\path[->, >=latex] (v7) edge (v6);
			\path[->, >=latex] (v8) edge (v7);
			\path[->, >=latex] (v9) edge (v5);
			\path[->, >=latex] (v10) edge (v5);
			\path[->, >=latex] (v11) edge (v10);
			\path[->, >=latex] (v12) edge (v11);
			\path[dotted] (v16) edge (v12);
		\end{tikzpicture}
	}
 
	\subcaptionbox{(${\mathcal G}, \pi_0$)=SHRINKING (STG, ${\mathcal P}({\mathcal Z}^{*})$), where  $v_{5,7}$, $v_{4,8}$ and $v_{9,10}$ are obtained by shrinking $\{v_5,v_7\}$, $\{v_4,v_8\}$ and $\{v_9,v_{10}\}$, respectively. Moreover, $\pi_0=\{\{v_1,v_{4,8}\},\{v_2,v_{5,7}\},\{v_3,v_6,v_{9,10}\},\{v_{11},\ldots,v_{16}\}\}$. \label{1-b}}
	{
		\begin{tikzpicture}[->,thick]
			\tikzstyle{every node}=[draw,circle,radius=0.2mm,inner sep=0pt,minimum size=1em];
			\def \n {12}
			\def \radius {0.6cm}
			\def \margin {45} 
			\node (v2)[] at ({360/6 * (1 - 1)+30}:\radius) {$2$};
			\node (v1)[fill=d2] at ({360/6 * (2 - 1)+30}:\radius) {$1$};
			\node (v6)[fill=d1] at ({360/6 * (3 - 1)+30}:\radius) {$6$};
			\node (v5)[fill=d4,label={[shift={(0.0,-0.9)}]$v_{5,7}$}] at ({360/6 * (4 - 1)+30}:\radius) {};
			\node (v4)[fill=d2,label={[shift={(0.2,-0.9)}]$v_{4,8}$}] at ({360/6 * (5 - 1)+30}:\radius) {};
			\node (v3)[fill=d1] at ({360/6 * (6 - 1)+30}:\radius) {$3$};
			\node (v10)[fill=d1,label={[shift={(0.0,-0.2)}]$v_{9,10}$}] at (-1.8,-0.3) {};
			\node (v11)[fill=c2] at (-3, -0.3) {$11$};
			\node (v12)[fill=c2] at (-4.2,-0.3) {$12$};
			\node (v16)[fill=c2] at (-5.4,-0.3) {$16$};
			
			\draw[->, >=latex] ({360/6 * (1 - 1)-30+\margin}:\radius)
			arc ({360/6 * (1 - 1)-30+\margin}:{360/6 * (1)-30-\margin}:\radius);
			\draw[->, >=latex] ({360/6 * (2- 1)-30+\margin}:\radius)
			arc ({360/6 * (2 - 1)-30+\margin}:{360/6 * (2)-30-\margin}:\radius);
			\draw[->, >=latex] ({360/6 * (3 - 1)-30+\margin}:\radius)
			arc ({360/6 * (3 - 1)-30+\margin}:{360/6 * (3)-30-\margin}:\radius);
			\draw[->, >=latex] ({360/6 * (4 - 1)-30+\margin}:\radius)
			arc ({360/6 * (4 - 1)-30+\margin}:{360/6 * (4)-30-\margin}:\radius);
			\draw[->, >=latex] ({360/6 * (5 - 1)-30+\margin}:\radius)
			arc ({360/6 * (5 - 1)-30+\margin}:{360/6 * (5)-30-\margin}:\radius);
			\draw[->, >=latex] ({360/6 * (6 - 1)-30+\margin}:\radius)
			arc ({360/6 * (6 - 1)-30+\margin}:{360/6 * (6)-30-\margin}:\radius);
			\path[->, >=latex] (v10) edge (v5);
			\path[->, >=latex] (v11) edge (v10);
			\path[->, >=latex] (v12) edge (v11);
			\path[dotted] (v16) edge (v12);
		\end{tikzpicture}
	}

	\subcaptionbox{${\mathcal G}_1$, obtained from ${\mathcal G}$ by replacing the original 6-cycle with 3-cycle $\{\{v_1,v_{4,8}\}$, $\{v_3,v_6\}$ and $\{v_2,v_{5,7}\}\}$.\label{1-d}}
	{
		\begin{tikzpicture}[->,thick]
			\tikzstyle{every node}=[draw,circle,radius=0.2mm,inner sep=0pt,minimum size=1em];
			\def \n {12}
			\def \radius {0.4cm}
			\def \margin {20} 
			\node  (v1)[fill=d2,label={[shift={(0.6,-0.6)}]$v_{1,4,8}$}] at ({360/3 * (1 - 1)+30}:\radius) {};
			\node (v3)[fill=d1,label={[shift={(-0.2,-0.2)}]$v_{3,6}$}] at ({360/3 * (2 - 1)+30}:\radius) {};
			\node (v2)[fill=d4,label={[shift={(0.2,-1)}]$v_{2,5,7}$}] at ({360/3 * (3 - 1)+30}:\radius) {};
			\node (v10)[fill=d1,label={[shift={(0.1,-1)}]$v_{9,10}$}] at (-1.5,-0.37) {};
			\node (v11)[fill=c2] at (-2.5, -0.37) {$11$};
			\node (v12)[fill=c2] at (-3.5,-0.37) {$12$};
			\node (v16)[fill=c2] at (-4.5,-0.37) {$16$};
			
			\draw[->, >=latex] ({30-\margin}:\radius)
			arc ({30-\margin}:{-90+\margin}:\radius);
			\draw[->, >=latex] ({-90-\margin}:\radius)
			arc ({-90-\margin}:{-210+\margin}:\radius);
			\draw[->, >=latex] ({150-\margin}:\radius)
			arc ({150-\margin}:{30+\margin}:\radius);
			\path[->, >=latex] (v10) edge (v2);
			\path[->, >=latex] (v11) edge (v10);
			\path[->, >=latex] (v12) edge (v11);
			\path[dotted] (v16) edge (v12);
		\end{tikzpicture}
	}
	\caption{Given a connected STG and a dual subspace ${\mathcal Z}^{*}$, the process of finding ${\mathcal P}(\overline{{\mathcal Z}^{*}})$, the coarsest equitable partition finer than ${\mathcal P}({\mathcal Z}^{*})$ is illustrated, where ${\mathcal P}({\mathcal Z}^{*})=\{\{v_1, v_4, v_8\},\{v_2,v_5,v_7\},\{v_3,v_6,v_9,v_{10}\},\{v_{11},\ldots ,v_{16}\}\}$. The dotted edge between $v_{16}$ and $v_{12}$ means the path from $v_{16}$ to $v_{12}$. We finally get ${\mathcal P}(\overline{{\mathcal Z}^{*}})=\{\{v_{1},v_{4},v_{8}\},\{v_3,v_6\},\{v_{2},v_{5},v_{7}\},\{v_9,v_{10}\},\{v_{11}\},$ $\ldots,\{v_{16}\}\}$, whose quotient digraph is illustrated in (c).\label{fig-5}}
\end{figure}

\begin{remr}
\label{r6}
 In subsections \ref{subsectionB-1} and \ref{cycle}, we present Algorithm \ref{a3.2.1} for determining ${\mathcal P}(\overline{{\mathcal Z}^{*}})$, the coarsest equitable partition finer than ${\mathcal P}({\mathcal Z}^{*})$, where ${\mathcal Z}^{*}$ is a given dual subspace. From Algorithm \ref{a3.2.1}, we can observe that the resulting quotient digraph corresponding to ${\mathcal P}(\overline{{\mathcal Z}^{*}})$ is either the result of a SHRINKING operation or a directed cycle (as specified in line 8 of Algorithm \ref{a3.2.1}). Consequently, we can conclude that the resulting quotient digraph satisfies two conditions: 1. The in-degree of each vertex does not exceed the cardinality of ${\mathcal P}({\mathcal Z}^{*})$. 2. If a vertex has multiple in-neighbors, then these in-neighbors belong to different cells in $\mathcal{P}(\mathcal{Z}^*)$.
\end{remr}

\section{Construction of observability outputs for a given BN}\label{IV}

Observability is a fundamental property in control theory. It provides the foundation for numerous related control problems, including state estimation, identification, disturbance decoupling, controller synthesis, etc. In the context of Boolean Control Networks (BCNs), there are primarily three methods for verifying observability, which are mathematically equivalent: Moore's partition-based method \citep{moor56}, the observability-graph method \citep{Zhang2014ObservabilityofBCNCCC,Zhang2016ObservabilityofBCN}, and an algebraic-variety-based method \citep{Li2015ControlObservaBCN}. 
For a comprehensive and recent survey on the observability of BCNs, readers are referred to \citep{Zhang2023SurveyObservabilityBCN}. 
Among these three methods, the second is the most widely used. By using the observability graph or its adjacency matrix, further results on the observability of BCNs were obtained
\citep{Cheng2016NoteonObservabilityBCN,Zhu2018ObservabilityBCN,Cheng2018ObservabilityBCNSetControl,Guo2018ObservabilityBCN,Zhang2021ObservabilityBCNSetControl};
observability verification results were extended from BCNs to probabilistic BCNs \citep{Zhou2019ObservabilityPBN,Yu2021ObservabilityBooleanNetwork} and stochastic labeled graphs \citep{Zhu2023ObserDetecStcoLabelGraph}; minimal observability \citep{Liu2022MinimalObservabilityBCN,xu24} and
observability perturbation analysis \citep{Wang2021PerburbationAnaObservabilityBCN} in BCNs were also investigated, just to name a few. Moreover, a slight variant of the observability graph was used to verify reconstructibility (also called detectability) of BCNs \citep{Zhang2016WPGRepresentationReconBCN} and of singular BCNs \citep{Li2020ReconstructibilitySBCN}.

Moore's partition was used to verify observability of BCNs in \citep{Fornasini2013ObservabilityReconstructibilityofBCN,guo18}.
Coincidentally, the method used for solving the disturbance decoupling problem of BCNs almost coincides with Moore's partition \citep{Li2012ControllabilitySwitchBCN}; the $M$-invariant dual subspaces of BNs generated by a set of output functions \citep{che21} also coincide with Moore's partition.
\subsection{Unobservable subspaces and the smallest \texorpdfstring{$M$}--invariant dual subspaces}
A BN is described as the following algebraic form \citep{che11}
\begin{equation}\label{BN-out}
\left\{
\begin{aligned}
{\bar x}(t+1)&=M{\bar x}(t), &{\bar x}(t)\in\Delta_{2^n}, \\
{\bar y}(t)&=E{\bar x}(t),&{\bar y}(t)\in\Delta_{2^q},
\end{aligned}
\right.
\end{equation}
where $M\in {\mathcal L}_{2^n\times 2^n}$ and $E\in {\mathcal L}_{2^q\times 2^n}$ are the state-transition and output matrices, respectively. The solution to BN (\ref{BN-out}) with initial state ${\bar x}_0\in\Delta_{2^n}$ at setp $t$ is denoted by ${\bar x}(t;{\bar x}_0)$. The output is denoted by ${\bar y}(t;{\bar x}_0)$, that is, ${\bar y}(t;{\bar x}_0)=E{\bar x}(t;{\bar x}_0)$. For convenience, we define  $\mathbf{y}(t ; x_{0}):={\bar y}(0 ; {\bar x}_{0}) \ltimes {\bar y}(1 ; {\bar x}_{0}) \ltimes \cdots \ltimes {\bar y}(t ; {\bar x}_{0})$.

Two distinct initial states ${\bar x}_0$ and ${\bar x}'_0$ are said to be distinguishable if there exists a positive integer $t$ such that $\mathbf{y}(t; {\bar x}_0)\neq \mathbf{y}(t; {\bar x}'_0)$. BN (\ref{BN-out}) is said to be observable if any two distinct initial states are distinguishable.
		
Denote $\mathcal{O}_r^*:=E*(EM)*\cdots *(EM^{r-1}).$ Then, $\mathbf{y}(r ; x_{0})=\mathcal{O}_r^*{\bar x}_0$.
Let the observability matrix be $\mathcal{O}_{r_0}^*$, where $r_0\!=\!\min\!\left\{\!r\!\mid\! \operatorname{rank}\left(\mathcal{O}_{r}^{*}\right)=\operatorname{rank}\left(\mathcal{O}_{r+1}^{*}\right)\!\right\}$.

\begin{lemm}[\!\!\citep{guo18}]\label{th4.1}
In BN $(\ref{BN-out})$, two distinct states ${\bar x}_0$ and ${\bar x}'_0$ are distinguishable iff $\mathcal{O}_{r_0}^*{\bar x}_0\neq\mathcal{O}_{r_0}^*{\bar x}'_0$. Moreover, BN $(\ref{BN-out})$ is observable iff no two columns of $\mathcal{O}_{r_0}^*$ are identical.
\end{lemm}
		
\begin{remr}
It can be seen that Lemma \ref{th4.1} is a special case of \citep[Theorem 6]{moor56}. Theorem 6 of \citep{moor56} was briefly restated in \citep[Remark 4.1]{zhang20} and \citep[Theorem  6]{Zhang2023SurveyObservabilityBCN}.
\end{remr}	
\begin{them}
\label{ovservability}
For BN $(\ref{BN-out})$, let $\overline{{\mathcal Z}^{*}}$ be the smallest $M$-invariant dual subspace generated by ${\mathcal Z}^{*}:={\mathcal F}_{\ell}\{y(t)\}$, where $y(t)\sim {\bar y}(t)$. Two states are distinguishable iff they are in different cells of ${\mathcal P}(\overline{{\mathcal Z}^{*}})$. Moreover, BN (\ref{BN-out}) is observable iff ${\mathcal P}(\overline{{\mathcal Z}^{*}})$ is trivial.
\end{them}
\begin{proof}
Based on Algorithm 3.11 in \citep{che21} and the definition of $r_0$, the observability matrix $\mathcal{O}_{r_0}^*$ is exactly the structure matrix of $\overline{{\mathcal Z}^{*}}$. By Theorem \ref{l3.1}, two states ${\bar x}_0$ and ${\bar x}'_0$ are in different cells of ${\mathcal P}(\overline{{\mathcal Z}^{*}})=\pi_{\mathcal{O}_{r_0}^*}$ iff $\mathcal{O}_{r_0}^*{\bar x}_0\neq\mathcal{O}_{r_0}^*{\bar x}'_0$. It follows that two states are distinguishable iff they are in different cells of ${\mathcal P}(\overline{{\mathcal Z}^{*}})$. And BN (\ref{BN-out}) is observable iff ${\mathcal P}(\overline{{\mathcal Z}^{*}})$ is trivial.
\end{proof}	
		
\begin{dfn}
For BN $(\ref{BN-out})$, the smallest $M$-invariant dual subspace generated by ${\mathcal F}_{\ell}\{y(t)\}$, where  $y(t)\sim {\bar y}(t)$,
is called the unobservable subspace of $(\ref{BN-out})$.
\end{dfn}
		
\subsection{Construction of observable output functions}
Utilizing the complete structural characterization of the smallest $M$-invariant dual subspaces generated by a set of Boolean functions, as provided in subsection \ref{III-B}, we can construct output functions that make a given BN observable.

To facilitate the statement of the subsequent theorem, we introduce the following definition: Given a partition $\pi$ of the vertex set, two $l$-cycles ${\mathcal C}_1$ and ${\mathcal C}_2$ are said to be shrinkable if there exists a pair of vertices $u \in V({\mathcal C}_1)$ and $v \in V({\mathcal C}_2)$ such that $N_{out}(u, i)$ and $N_{out}(v, i)$ are in the same cell of $\pi$ for all $i \in [0;l-1]$.


\begin{them}
\label{t4.3}
Suppose that BN $(\ref{BN-out})$ has output function $y(t)$. If ${\mathcal Z}^{*}={\mathcal F}_{\ell}\{y(t)\}$ satisfies the following conditions:
\begin{itemize}
\item[(i)] 
In STG ${\mathcal G}$, the in-neighbors of any vertex are in distinct cells of ${\mathcal P}({\mathcal Z}^{*})$;
\item[(ii)]
For each cycle in ${\mathcal G}$, there exists a vertex $v$ within the cycle such that $v$ is not in the same cell of ${\mathcal P}({\mathcal Z}^{*})$ with any other vertex of this cycle;
\item[(iii)]
In the partition ${\mathcal P}({\mathcal Z}^*)$, any two distinct cycles of equal length are not shrinkable; 
\item[(iv)] All vertices incident with loops are in different cells of ${\mathcal P}({\mathcal Z}^{*})$.
\end{itemize}
then $(\ref{BN-out})$ is observable.
\end{them}
		
\begin{proof}

According to Theorem \ref{ovservability}, we have proved that system $(\ref{BN-out})$ is observable if ${\mathcal P}(\overline{{\mathcal Z}^{*}})$ is trivial. We now provide separate proofs of Theorem \ref{t4.3} depending on whether STG ${\mathcal G}$ is connected or not.

(1) ${\mathcal G}$ is connected.

(1.1) If ${\mathcal G}$ contains a loop, the condition (i) means that ${\mathcal G}$ is the quotient digraph of ${\mathcal P}(\overline{{\mathcal Z}^{*}})$ according to Remark \ref{r6}. Equivalently, ${\mathcal P}(\overline{{\mathcal Z}^{*}})$ is trivial.

(1.2) Consider the case where ${\mathcal G}$ contains a cycle ${\mathcal C}$. Let $\pi'_0$ denote the partition induced by $V({\mathcal C})$ from ${\mathcal P}({\mathcal Z}^*)$.
As stated in condition (ii), there exists a vertex in $V({\mathcal C})$ that is partitioned into different cells of $\pi'_0$ from any other vertex in $V({\mathcal C})$. Based on Theorem \ref{l3.10}, the equitable partition of ${\mathcal C}$ generated by $\pi'_0$ is trivial. In this case, condition (i) means that ${\mathcal G}$ is the quotient digraph of ${\mathcal P}(\overline{{\mathcal Z}^{*}})$, according to Theorem \ref{the-5}. In other words, ${\mathcal P}(\overline{{\mathcal Z}^{*}})$ is trivial.

(2) If ${\mathcal G}$ is not connected, we can infer from the preceding proof that any two vertices in each component of ${\mathcal G}$ belong to different cells of ${\mathcal P}(\overline{{\mathcal Z}^{*}})$. We now prove that every vertex is distinguishable from all vertices in other components.

(2.1) Condition (iv) asserts that all vertices incident with loops are in distinct cells of ${\mathcal P}(\overline{{\mathcal Z}^{*}})$. 


(2.2) From the preceding proof, we can conclude that for any cycle of length $l$ in the digraph, its vertices belong to different cells of ${\mathcal P}(\overline{{\mathcal Z}^{*}})$. Thus, in the quotient digraph corresponding to ${\mathcal P}(\overline{{\mathcal Z}^{*}})$, the vertices resulting from shrinking the cells containing these vertices in the $l$-cycle can induce an $l$-cycle. This implies that vertices in cycles with different lengths cannot be in the same cell of ${\mathcal P}(\overline{{\mathcal Z}^{*}})$. 

We now prove that any two vertices from two cycles of equal length $l$ cannot be in the same cell of ${\mathcal P}(\overline{{\mathcal Z}^{*}})$. Let ${\mathcal C}_1$ and ${\mathcal C}_2$ be two distinct $l$-cycles. Suppose, by way of contradiction, that there exists a pair of vertices $u\in V({\mathcal C}_1)$ and $v\in V({\mathcal C}_2)$ contained in the same cell of ${\mathcal P}(\overline{{\mathcal Z}^{*}})$. By the definition of the equitable partition, $N_{out}(u, i)$ and $N_{out}(v, i)$ must be in the same cell of ${\mathcal P}(\overline{{\mathcal Z}^{*}})$ for all $i \in [0;l-1]$. Since ${\mathcal P}({\mathcal Z}^{*})\preceq {\mathcal P}(\overline{{\mathcal Z}^{*}})$, this implies that $u$ and $v$ are in the same cell of ${\mathcal P}({\mathcal Z}^{*})$. However, this contradicts condition (iii). Thus, the vertices in cycles with the same length cannot be in the same cell of ${\mathcal P}(\overline{{\mathcal Z}^{*}})$. 

In conclusion, we have demonstrated that all vertices in cycles (including loops, which are cycles of length 1) are in distinct cells of ${\mathcal P}(\overline{{\mathcal Z}^{*}})$. Consequently, these vertices are mutually distinguishable.

(2.3) We now prove that any two states $u$ and $v$ in different components of ${\mathcal G}$ are distinguishable from each other. 
According to \citet{che12}, any trajectory in a BN eventually converges to a directed cycle (including loops, which are cycles of length 1). Consider the BN (\ref{BN-out}) with initial states $u\in\Delta_{2^n}$ and $v\in\Delta_{2^n}$. For sufficiently large step $t$, the corresponding solutions ${\bar x}(t;u)$ and ${\bar x}(t;v)$ necessarily are states in cycles. As we have previously shown, these cycle states are distinguishable. Consequently, we conclude that states $u$ nor $v$ are distinguishable.

\end{proof}
		
\begin{exa}{\label{e3.9}}
{\color{red}{Consider the BN (\ref{ep-new2-3.8}) in Example \ref{e3.2}.}} Its STG ${\mathcal G}$ is shown in Fig. \ref{fig3.1}. Utilizing Theorem \ref{t4.3}, we construct an observable output function $y(t)$. To satisfy condition (i) of Theorem \ref{t4.3}, we need at least 9 cells to partition the in-neighbors of $\delta_{32}^{32}$ into different cells. An observable output matrix is
$$
\setlength{\arraycolsep}{1.8pt}
\begin{array}{llllllllllllllllllllllllllllllllllll}
\!\!E=\delta_{16}[\!\!&9\!\!&2\!\!&2\!\!&3\!\!\!\!\!\!\!\!&3 &4 &4 &5
&1 &1 &2 &2 &2 &2 &1 &1\\
&5 &6 &6 &7 &7 &8 &8 &1
&1 &1 &1  &1 &1  &1 &2 &1].
\end{array}
$$

In BN
\begin{equation}{\label{3.8-1}}
\left\{
\begin{aligned}
{\bar x}(t+1)&=M{\bar x}(t),\\
{\bar x}(t)&=E{\bar x}(t),
\end{aligned}
\right.
\end{equation}
we can determine the initial state of any given output sequence. That is, this BN is observable.

\begin{figure}
\centering
{\begin{tikzpicture}[->,thick]
\tikzstyle{every node}=[draw,circle,radius=0.8mm,inner sep=0pt,minimum size=1.5em];
\node (v1) [fill=c9]at (1,0.8){$1$};
\node (v2)[fill=c2] at (6,2.4){$2$};
\node (v3)[fill=c2] at (1,1.6){$3$};
\node (v4)[fill=c3] at (6,1.6){$4$};
\node (v5)[fill=c3] at (1,2.4){$5$};
\node (v6)[fill=c4] at (6,0.8){$6$};
\node (v7)[fill=c4] at (1,-0.8){$7$};
\node (v8)[fill=c5] at (6,0){$8$};
\node (v9)[fill=c1] at (4,2.4){$9$};
\node (v10)[fill=c1] at (5,3.2){$10$};
\node (v11)[fill=c2]  at (4,1.6){$11$};
\node (v12)[fill=c2]  at (5,2.4){$12$};
\node (v13)[fill=c2]  at (3,-0.8){$13$};
\node (v14)[fill=c2]  at (3,2.4){$14$};
\node (v15)[fill=c1]  at (3,0){$15$};
\node (v16)[fill=c1]  at (3,3.2){$16$};
\node (v17)[fill=c5]  at (1,-1.6){$17$};
\node (v18)[fill=c6]  at (6,-0.8){$18$};
\node (v19)[fill=c6]  at (1,-2.4){$19$};
\node (v20)[fill=c7]  at (6,-1.6){$20$};
\node (v21)[fill=c7]  at (1,-3.2){$21$};
\node (v22)[fill=c8]  at (6,-2.4){$22$};
\node (v23)[fill=c8]  at (1,3.2){$23$};
\node (v24)[fill=c1]  at (7,0.4){$24$};
\node (v25)[fill=c1]  at (4,0){$25$};
\node (v26)[fill=c1]  at (5,1.6){$26$};
\node (v27)[fill=c1]  at (1,0){$27$};
\node (v28)[fill=c1]  at (5,0){$28$};
\node (v29)[fill=c1]  at (0,0){$29$};
\node (v30)[fill=c1]  at (3,1.6){$30$};
\node (v31)[fill=c2]  at (7,-2){$31$};
\node (v32)[fill=c1]  at (2,0){$32$};

\path (v29) edge (v27);
\path (v1) edge (v32);
\path (v3) edge (v32);
\path (v5) edge (v32);
\path (v7) edge (v32);
\path (v17) edge (v32);
\path (v19) edge (v32);
\path (v21) edge (v32);
\path (v23) edge (v32);
\path (v27) edge (v32);
\path (v32) edge (v15);
\path (v14) edge (v9);
\path (v16) edge (v9);
\path (v30) edge (v11);
\path (v13) edge (v25);
\path (v15) edge (v25);
\path (v9) edge (v26);
\path (v11) edge (v26);
\path (v25) edge (v28);
\path (v10) edge (v2);
\path (v12) edge (v2);
\path (v26) edge (v4);
\path (v28) edge (v8);
\path (v2) edge (v24);
\path (v4) edge (v24);
\path (v6) edge (v24);
\path (v8) edge (v24);
\path (v18) edge (v24);
\path (v20) edge (v24);
\path (v22) edge (v24);
\draw (v24) to [out=45,in=315,looseness=6] (v24);
\draw (v31) to [out=45,in=315,looseness=6] (v31);
\end{tikzpicture}
}
\caption{STG of BN (\ref{ep-new2-3.8}). For the partition $\pi=\{\{x|Ex=\delta_{9}^{i}\}, i=1,2\ldots,9\}$, vertices in the same cell are assigned with the same colour. }
\label{fig3.1}
\end{figure}
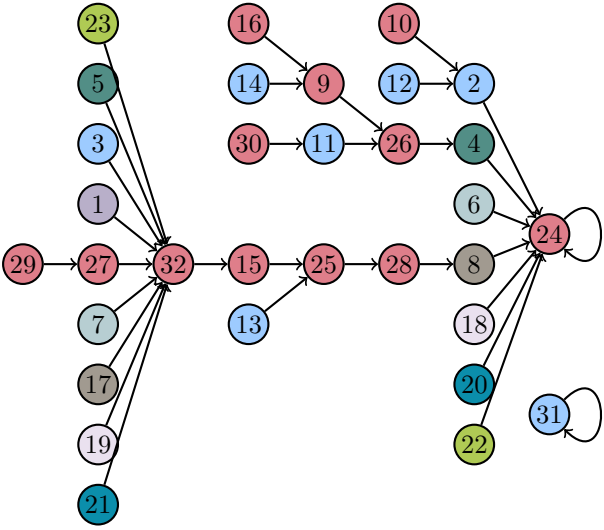
\end{exa}

\section{Conclusions}

In this paper, we constructed a bijection between dual subspaces and partitions of the state-transition graph of a BN, where these partitions can be reduced by equivalence relations. Furthermore, we proved that a dual subspace is $M$-invariant iff the corresponding partition is equitable. Thus, we can describe the dynamics of the equivalence classes obtained from an $M$-invariant dual subspace using the quotient digraph induced by the corresponding equitable partition. On the other hand, with the help of this bijection, we thoroughly characterised the structures of the smallest $M$-invariant dual subspaces generated by a set of Boolean functions. We proved that a BN with given output functions is observable iff the partition corresponding to the smallest $M$-invariant dual subspace containing the output functions (defined as the unobservable subspace) is trivial. We obtained a method for constructing observable output functions based on the structural characterization.

\bibliographystyle{model5-names}
\bibliography{cas-refs}

\end{document}